\documentclass[letterpaper,11pt]{article}

\usepackage{ifthen}

\usepackage{times} 
\usepackage{fullpage} 
\usepackage{amsmath} 
\usepackage{amssymb} 
\usepackage{amsthm} 
\usepackage{bbm} 

\usepackage{tabularx}

\usepackage{graphicx} 

\usepackage{ifpdf}
\ifpdf
\usepackage[
  pdftex,
  bookmarks,
  pagebackref,
  plainpages=false, 
  pdfpagelabels=true 
]{hyperref}
\else \fi

\newcommand{\inner}[2]{\langle #1 , #2\rangle}
\newcommand{\Inner}[2]{\left\langle #1 , #2\right\rangle}
\newcommand{\defeq}{\stackrel{\smash{\textnormal{\tiny def}}}{=}}

\newcommand{\cls}[1]{\mathrm{#1}}

\newtheorem{theorem}{Theorem}
\newtheorem{lemma}[theorem]{Lemma}
\newtheorem{corollary}{Corollary}[theorem]
\newtheorem{proposition}[theorem]{Proposition}

\theoremstyle{definition}
\newtheorem{defn}[theorem]{Definition}

\newtheorem{problem}{Problem}

\newcommand{\pa}[1]{(#1)}
\newcommand{\Pa}[1]{\left(#1\right)}

\newcommand{\Br}[1]{\left[#1\right]}
\newcommand{\set}[1]{\{#1\}}

\newcommand{\bra}[1]{\langle#1|}

\newcommand{\ket}[1]{|#1\rangle}

\DeclareMathOperator{\trace}{Tr}
\newcommand{\ptr}[2]{\trace_{#1}\pa{#2}}

\newcommand{\tinyspace}{\mspace{1mu}}
\newcommand{\abs}[1]{|\tinyspace#1\tinyspace|}

\newcommand{\norm}[1]{\lVert\tinyspace#1\tinyspace\rVert}
\newcommand{\Norm}[1]{\left\lVert\tinyspace#1\tinyspace\right\rVert}

\newcommand{\tnorm}[1]{\norm{#1}_{\trace}}
\newcommand{\Tnorm}[1]{\Norm{#1}_{\trace}}

\newcommand{\identity}{\mathbbm{1}}
\newcommand{\idsup}[1]{\identity_{#1}}

\def\ot{\otimes}

\def\cA{\mathcal{A}}

\def\cM{\mathcal{M}}

\def\cV{\mathcal{V}}
\def\cW{\mathcal{W}}
\def\cX{\mathcal{X}}
\def\cY{\mathcal{Y}}
\def\cZ{\mathcal{Z}}

\def\bA{\mathbf{A}}

\def\bP{\mathbf{P}}

\def\bY{\mathbf{Y}}

\def\sM{\mathsf{M}}

\def\sV{\mathsf{V}}
\def\sW{\mathsf{W}}

\def\sZ{\mathsf{Z}}

\def\bbM{\mathbb{M}}

\newcommand{\compileCC}{0}

\begin{document}

\title{
Parallel approximation of min-max problems
}

\author{
  Gus Gutoski$^*$ $\qquad$ Xiaodi Wu$^\dagger$ \\[3mm]
  {\small\it
  \begin{tabular}{c}
    {\large$^*$}Institute for Quantum Computing and School of Computer Science \\
    University of Waterloo,
    Waterloo, Ontario, Canada \\[1mm]
    {\large$^\dagger$}Department of Electrical Engineering and Computer Science \\
    University of Michigan,
    Ann Arbor, Michigan, USA
  \end{tabular}
  }
}

\date{December 7, 2012}

\maketitle

\begin{abstract}

This paper presents an efficient parallel approximation scheme for a new class of min-max problems.
The algorithm is derived from the matrix multiplicative weights update method and can be used to find near-optimal strategies for competitive two-party classical or quantum interactions in which a referee exchanges any number of messages with one party followed by any number of additional messages with the other.
It considerably extends the class of interactions which admit parallel solutions, demonstrating for the first time the existence of a parallel algorithm for an interaction in which one party reacts adaptively to the other.

As a consequence, we prove that several competing-provers complexity classes collapse to $\cls{PSPACE}$ such as $\cls{QRG}(2)$, $\cls{SQG}$ and two new classes called $\cls{DIP}$ and $\cls{DQIP}$.
A special case of our result is a parallel approximation scheme for a specific class of semidefinite programs whose feasible region consists of lists of semidefinite matrices that satisfy a transcript-like consistency condition.
Applied to this special case, our algorithm yields a direct polynomial-space simulation of multi-message quantum interactive proofs resulting in a first-principles proof of $\cls{QIP}=\cls{PSPACE}$.

\end{abstract}

\section{Introduction}
\label{sec:intro}

This paper presents a parallel approximation scheme for a new class of min-max problems with applications to classical and quantum zero-sum games and interactive proofs.
In order to describe this class of min-max problems let us begin by considering a semidefinite program (SDP) of the form
\begin{gather}
  \label{eq:SDP}
  \begin{aligned}
    \textrm{minimize} \qquad & \ptr{}{X_kP} \\
    \textrm{subject to} \qquad &
    \ptr{\bbM_n}{X_{i+1}} = \Phi_i(X_i) \textrm{ for $i=1,\dots,k-1$} \\
    & \ptr{}{X_1} = 1 \\
    & 0\preceq X_1,\dots,X_k\in\bbM_{mn}
  \end{aligned}
\end{gather}
Here $\bbM_d$ denotes the space of all $d\times d$ complex matrices and $\trace_{\bbM_n}$ is the \emph{partial trace}---the unique linear map from matrices to matrices satisfying
\[ \trace_{\bbM_n}:\bbM_{mn}\to\bbM_m:A\ot B\mapsto\ptr{}{B}A \]
for every choice of $A\in\bbM_m$ and $B\in\bbM_n$.
An SDP \eqref{eq:SDP} is specified by arbitrary choices of a positive semidefinite matrix $P\in\bbM_{mn}$ with $\norm{P}\leq 1$ and completely positive and trace-preserving linear maps
\( \Phi_1,\dots,\Phi_{k-1}:\bbM_{mn}\to\bbM_m. \)
(A linear map $\Phi$ is \emph{positive} if $\Phi(X)\succeq 0$ whenever $X\succeq 0$.
Such a map is \emph{completely positive} if $\Phi\ot\idsup{\bbM_d}$ is positive for every positive integer $d$.)

Let $\bA$ denote the feasible region of the SDP \eqref{eq:SDP} (which is always non-empty) and let $\bP\subset\bbM_{mn}$ be a non-empty compact convex subset of positive semidefinite matrices having operator norm at most 1.
We are concerned with the following min-max problem, which is a generalization of the SDP \eqref{eq:SDP}:
\begin{equation}
  \lambda(\bA,\bP) \defeq \min_{(X_1,\dots,X_k)\in\bA} \ \max_{P\in\bP} \ \ptr{}{X_kP} \label{eq:min-max}
\end{equation}
The ordering of minimization and maximization is immaterial, as implied by well-known extensions of von Neumann's Min-Max Theorem \cite{vonNeumann28,Fan53} given the fact that $\bA,\bP$ are convex compact sets and $\ptr{}{X_kP}$ is a bilinear form over the two sets.

Our main result is an efficient parallel oracle-algorithm for finding approximate solutions to the min-max problem \eqref{eq:min-max} and for approximating the quantity $\lambda(\bA,\bP)$, given an oracle for optimization over the set $\bP$.
We also describe parallel implementations of this oracle for certain sets $\bP$, yielding an unconditionally efficient parallel approximation scheme for the min-max problem \eqref{eq:min-max} for those choices of $\bP$.
This result is stated formally below as Theorem \ref{thm:main-result}.
Before stating this theorem let us clarify terminology.

\subsection{Review of parallel computation, formal statement of results}

Recall that a \emph{parallel algorithm} is described by a family of logarithmic-space uniform Boolean circuits.
The uniformity constraint ensures that the size of each circuit in the family scales as a polynomial in the bit length of the input, and therefore the family represents a polynomial-time computation.
Boolean circuits are an ideal model of parallel computation because computational activity can occur concurrently at many different gates in the circuit.
Indeed, the run time of a parallel algorithm is determined by the \emph{depth} of its circuits, which might be much smaller than the total \emph{size} of its circuits.

A parallel algorithm is said to be \emph{efficient} if the depth of its circuits (and therefore the run time of the algorithm) scales as a polynomial in the \emph{logarithm} of the bit length of the input.
The complexity class $\cls{NC}$ consists of those functions which can be computed by efficient parallel algorithms.
Efficient parallel algorithms are sometimes called ``NC algorithms'' or ``NC computations.''
The reader is referred to \cite{Papadimitriou94} for an accessible introduction to parallel computation.

An \emph{oracle-algorithm} is an algorithm endowed with the ability to get instantaneous answers to questions that fall within the scope of some specific \emph{oracle}.
In our case, we assume an oracle for optimization over $\bP$, which instantly solves problems of the form
\begin{problem}
  [Optimization over $\bP$]
  \label{problem:oracle}
  \ \\[1mm]
  \begin{tabularx}{\textwidth}{lX}
    \emph{Input:} &
    A matrix $X\succeq 0$ with $\ptr{}{X}= 1$ and an accuracy parameter $\delta>0$.
    \\[1mm]
    \emph{Output:} &
    A near-optimal element $P^\star\in\bP$ such that $\ptr{}{XP^\star}\geq \ptr{}{XP}-\delta$ for all $P\in\bP$.
  \end{tabularx}
\end{problem}
An oracle is incorporated into the circuit model of computation by supplementing a standard gate set (such as $\set{\mathrm{AND}, \mathrm{OR}, \mathrm{NOT}}$) with a special \emph{oracle gate}.
This oracle gate has many input bits (describing the question) and many output bits (describing the answer).
As with standard gates, each oracle gate contributes unit cost to circuit size and run time.

An \emph{approximation scheme} refers to an algorithm that computes one or more quantities to a given precision $\delta$ and whose run time is efficient for each fixed choice of $\delta > 0$ but does not necessarily scale well with $\delta$.
In the circuit model (and other models, too) this property is encapsulated by defining the underlying problem so that the accuracy parameter $\delta=1/s$ is specified in unary as $1^s$, thus forcing the bit length of the input to be proportional to $1/\delta$ instead of $1/\log(\delta)$.
The choice to specify the accuracy parameter in unary allows parallel approximation schemes to be described neatly by log-space uniform circuits with polylog depth.

The following is a formal statement of the problem solved by our algorithm.
\begin{problem}
  [Approximation of $\lambda(\bA,\bP)$]
  \label{problem:lambda}
  \ \\[1mm]
  \begin{tabularx}{\textwidth}{lX}
    \emph{Input:} &
    Completely positive and trace-preserving linear maps $\Phi_1,\dots,\Phi_{k-1}$ specifying the feasible region $\bA$ of an SDP of the form \eqref{eq:SDP}.
    An accuracy parameter $\delta>0$.
    \\[1mm]
    \emph{Oracle:} &
    Optimization over $\bP$ (Problem \ref{problem:oracle}).
    \\[1mm]
    \emph{Output:} &
    Near-optimal elements $(X_1^\star,\dots,X_k^\star)\in\bA$ and $P^\star\in\bP$ such that
    \[
    \begin{aligned}
      \ptr{}{X_k^\star P} & \leq \lambda(\bA,\bP) + \delta \textrm{ for all $P\in\bP$} \\
      \ptr{}{X_k P^\star} & \geq \lambda(\bA,\bP) - \delta \textrm{ for all $(X_1,\dots,X_k)\in\bA$}
    \end{aligned}
    \]
    and a quantity $\tilde\lambda$ with $\abs{\tilde\lambda - \lambda(\bA,\bP)}\leq \delta$.
  \end{tabularx}
\end{problem}

The maps $\Phi_1,\dots,\Phi_{k-1}$ are linear maps from a complex vector space of dimension $(mn)^2$ to another complex space of dimension $m^2$.
As such, these maps can be represented by complex matrices of size $m^2\times(mn)^2$.
In both Problem \ref{problem:oracle} and Problem \ref{problem:lambda} it is assumed that the real and imaginary parts of each entry in each input matrix are represented as rational numbers expressed as the ratio of two $p$-bit integers written in binary for some $p$ that is promised to scale as a polynomial in the dimension $mn$.
(Indeed, it suffices for our purpose that $p$ scales \emph{logarithmically} with $mn$.)
As suggested previously, it is also assumed that the accuracy parameter $\delta$ is represented in unary.
These assumptions allow us to focus on the quantities $mn, k, 1/\delta$ as the dominating factors determining the run time of our parallel algorithm.
We may now state our main result.

\begin{theorem}[Main result]
\label{thm:main-result}
  There is a parallel oracle-algorithm for Problem \ref{problem:lambda} (Approximation of $\lambda(\bA,\bP)$) with run time bounded by a polynomial in $k$, $1/\delta$, and $\log (mn)$.
  This algorithm is efficient if $k,1/\delta$ are promised to scale as a polynomial in $\log (mn)$.
\end{theorem}

\subsection{Application: parallel approximation of semidefinite programs}
\label{sec:intro:sdp}

The SDP \eqref{eq:SDP} is recovered from \eqref{eq:min-max} in the special case where $\bP=\set{P}$ is a singleton set.
Thus, a special case of Theorem \ref{thm:main-result} is a parallel approximation scheme for SDPs of the form \eqref{eq:SDP}.

We restricted attention to SDPs for which $\norm{P},\ptr{}{X_1}\leq 1$ because this restriction does not interfere with our application to quantum interactive protocols and because the run time of our parallel algorithm scales polynomially with the largest eigenvalue of $P$ and with the trace of $X_1$, so it is only efficient when these quantities are bounded by a fixed polynomial in the logarithm of the bit length of the input $P,\Phi_1,\dots,\Phi_{k-1}$.
(In keeping with convention, one can think of these quantities as the \emph{width} of the SDPs we consider.
Our algorithm is efficient only for \emph{width-bounded} SDPs.)

It has long since been known that the problem of approximating the optimal value of an arbitrary SDP is logspace-hard for $\cls{P}$ \cite{Serna91,Megiddo92},
so there cannot be a parallel approximation scheme for \emph{all} SDPs unless $\cls{NC}=\cls{P}$.
The precise extent to which SDPs admit parallel solutions is not known.
This special case of our result adds considerably to the set of such SDPs, subsuming all prior work in the area at the time it was made public.
(Since that time parallel approximation schemes have been found for some SDPs of unbounded width that are not covered by our scheme \cite{JainY11,PengT12,JainY12}.)

Some of what is known about SDPs in this respect is inherited knowledge from linear programs (LPs).
For example, Luby and Nisan describe a parallel approximation scheme for so-called \emph{positive} LPs
of the form
\[ \textrm{minimize $xp^*$ subject to $Cx \geq q$ and $x \geq 0$} \]
where each entry of the matrix $C$ and vectors $p,q$ is a nonnegative real number \cite{LubyN93}.
Young provides a generalization of Luby-Nisan to arbitrary mixed packing and covering problems \cite{Young01}.
By contrast, Trevisan and Xhafa show that it is $\cls{P}$-hard to find \emph{exact} solutions
for positive LPs \cite{TrevisanX98}.\footnote{
  For clarification, a polynomial-time algorithm finds an \emph{exact} solution to an LP or SDP if it finds solutions that are within $\varepsilon$ of optimal in time polynomial in the bit length of $\varepsilon$---that is, $\log(1/\varepsilon)$.
  By contrast, an \emph{approximation scheme} for LPs or SDPs finds solutions that are within $\varepsilon$ of optimal with run time that depends super-polynomially in the bit length of $\varepsilon$---typically $1/\varepsilon$.
}

The notion of a positive instance of an LP can be generalized to SDPs as follows.
An SDP of the form
\[ \textrm{minimize $\ptr{}{XP}$ subject to $\Psi(X) \succeq Q$ and $X \succeq 0$} \]
is said to be \emph{positive} if $P,Q\succeq 0$ and $\Psi$ is a positive map.
Of course, $\cls{P}$-hardness of exact solutions for positive LPs implies $\cls{P}$-hardness of exact solutions for positive SDPs.
Jain and Watrous give a parallel approximation scheme for width-bounded positive SDPs \cite{JainW08}.
Subsequent improvements extend to all positive SDPs \cite{JainY11,PengT12}, and even to mixed packing and covering SDPs \cite{JainY12}.

The Jain-Watrous algorithm for positive SDPs is derived from a correspondence between positive SDPs and one-turn quantum games and can therefore be recovered as a special case of the work of the present paper.
In their proof of $\cls{QIP}=\cls{PSPACE}$, Jain \emph{et al.}\ give a parallel algorithm for a specific SDP based on quantum interactive proofs \cite{JainJ+11}.
It is not difficult to see that their SDP can be written in the form \eqref{eq:SDP} considered in the present paper.

As mentioned above, our algorithm is not efficient when used for SDPs of unbounded width, leaving the recent works of Jain and Yao \cite{JainY11,JainY12} and Peng and Tangwongsan \cite{PengT12} on mixed packing and covering SDPs as the only known parallel SDP approximation schemes that are not subsumed by the present work.
These recent works do not subsume our results, as neither the SDP instance used in Ref.\ \cite{JainJ+11} to prove $\cls{QIP}=\cls{PSPACE}$ nor its generalization \eqref{eq:SDP} in the present paper are mixed packing and covering SDPs.

\subsection{Application: interactive proofs with competing provers}
\label{sec:intro:dqip}

\subsubsection{Definitions}
\label{sec:intro:dqip:def}

An \emph{interactive proof with competing provers} consists of a conversation between a \emph{verifier} and two \emph{provers} regarding some input string $x$.
The verifier may use randomness, but must run in time that scales as a polynomial in the input length $|x|$; the provers are permitted unlimited computational power.
One of the provers---the \emph{yes-prover}---tries to convince the verifier to accept $x$, while the other---the \emph{no-prover}---tries to convince the verifier to reject $x$.
A decision problem $L$ is said to admit an interactive proof with competing provers with \emph{completeness} $c$ and \emph{soundness} $s$ if there exists $c,s$ with $c>s$ and a randomized polynomial-time verifier who meets the following conditions:
\begin{description}

\item[Completeness condition.]
If $x$ is a yes-instance of $L$ then the yes-prover can convince the verifier to accept with probability at least $c$ regardless of the no-prover's strategy.

\item[Soundness condition.]
If $x$ is a no-instance of $L$ then the no-prover can convince the verifier to reject with probability at least $1-s$ regardless of the yes-prover's strategy.

\end{description}

The completeness and soundness parameters $c,s$ need not be fixed constants, but may instead vary as a function of the input length $|x|$.
If these parameters are not specified then it is assumed that $L$ admits an interactive proof with competing provers for some choice of $c(|x|),s(|x|)$ for which there exists a polynomial-bounded function $p(|x|)$ such that $c-s\geq 1/p$.
The complexity class $\cls{RG}$ consists of all decision problems that admit interactive proofs with competing provers.
(The acronym $\cls{RG}$ stands for ``refereed games,'' a term inspired by the field of game theory).

Often in the study of interactive proofs the precise values of $c,s$ are immaterial because sequential repetition (or sometimes parallel repetition) can be used to transform any verifier for which $c-s\geq 1/p$ into another verifier for which $c$ tends toward one and $s$ tends toward zero exponentially quickly in the bit length of $x$.
(For example, sequential repetition followed by a majority vote can be used to reduce error for $\cls{RG}$.)
For this reason, it is typical to assume without loss of generality that $c,s$ are constants such as $2/3,1/3$ or that $c$ is exponentially close to one and $s$ is exponentially close to zero whenever it is convenient to do so.
However, it is not always clear that a given complexity class is robust with respect to the choice of $c,s$ so it is good practice to be as inclusive as possible when defining these classes.

Interesting subclasses of $\cls{RG}$ are obtained by placing restrictions upon the number and timing of messages in the interaction between the verifier and provers.
In this paper we introduce one such subclass based upon interactions of the following form:
\begin{enumerate}
\item The verifier exchanges several messages with only the yes-prover.
\item After processing this interaction with the yes-prover, the verifier exchanges several additional messages with only the no-prover.
\item
After further processing, the referee declares acceptance or rejection.
\end{enumerate}
Interactive proofs of this form shall be called \emph{double interactive proofs}: the verifier in such a protocol executes a standard single-prover interactive proof with the yes-prover followed by a second single-prover interactive proof with the no-prover.
The class of problems that admit double interactive proofs shall be called $\cls{DIP}$.

By contrast to $\cls{RG}$, it is not immediately clear that the definition of $\cls{DIP}$ is robust with respect to the choice of parameters $c,s$.
But it follows from our result that $\cls{DIP}$ is, in fact, robust with respect to the choice of $c,s$.
Also, whereas $\cls{RG}$ is trivially closed under complement, the protocol for double interactive proofs is asymmetric and so it is not immediately clear that $\cls{DIP}$ is closed under complement.
Again, it follows from our result that $\cls{DIP}$ is closed under complement.

Another example of an interesting subclass of $\cls{RG}$ is the family of \emph{bounded-turn} classes.
For each positive integer $k$ the class $\cls{RG}(k)$ consists of those problems that admit an interactive proof with competing provers in which the verifier exchanges no more than $k$ messages with each prover.
It is understood that
messages are exchanged with the provers in parallel
so that $\cls{RG}(k)$, like $\cls{RG}$, is trivially closed under complement.

\emph{Quantum} interactive proofs with competing provers are defined similarly except that the verifier is a polynomial-time quantum computer who exchanges quantum information with the provers.
The analogous complexity classes are denoted $\cls{QRG}$, $\cls{DQIP}$, and $\cls{QRG}(k)$.

\subsubsection{Prior work}

As noted in Refs.\ \cite{FeigenbaumK+95,FeigeK97}, the results of Koller and Meggido \cite{KollerM92} and Koller, Megiddo, and
von Stengel \cite{KollerMvS94} imply that $\cls{RG}\subseteq\cls{EXP}$.
The reverse containment was proven by Feige and Kilian \cite{FeigeK97}, yielding the characterization $\cls{RG}=\cls{EXP}$.
It was proven in Ref.\ \cite{GutoskiW07} that $\cls{QRG}\subseteq\cls{EXP}$, from which one obtains
\[ \cls{QRG}=\cls{RG}=\cls{EXP}, \]
which is the competing-provers version of the well-known collapse $\cls{QIP}=\cls{IP}=\cls{PSPACE}$ for single-prover interactive proofs \cite{LundF+92,Shamir92,JainJ+11}.

For bounded-turn classes, the results of Fortnow \emph{et al.}\ tell us that $\cls{RG}(1)$ is essentially a randomized version of $\cls{S^P_2}$ \cite{FortnowI+08}.
Feige and Kilian proved $\cls{RG}(2)=\cls{PSPACE}$ \cite{FeigeK97}.\footnote{
  The class we call $\cls{RG}(2)$ is called $\cls{RG}(1)$ by Feige and Kilian \cite{FeigeK97}.
  This conflict in notation stems from the fact that we measure the length of an interaction in \emph{turns} (\emph{i.e.}\ messages per prover), whereas those authors measure an interaction in \emph{rounds} of messages.
  This switch of notation was instigated by Jain and Watrous, who required a convenient symbol for one-turn interactions \cite{JainW08}.
}
For bounded-turn quantum classes, \cite{JainW08} proved $\cls{QRG}(1)\subseteq\cls{PSPACE}$.
The complexity of $\cls{QRG}(2)$ is an open question of \cite{JainJ+11} that is solved in the present paper.
The exact complexity of $\cls{RG}(k)$ and $\cls{QRG}(k)$ for all other $k$ is not known.

Bounded-turn double quantum interactive proofs have been studied previously under the name \emph{short quantum games}; the associated complexity class has been called $\cls{SQG}$.
In an effort to unify notation let $\cls{DQIP}(k,l)$ denote the class consisting of problems that admit a double quantum interactive proof with competing provers in which the verifier exchanges no more than $k$ messages with the yes-prover followed by no more than $l$ messages with the no-prover.
The class $\cls{SQG}$ was first defined in Ref.\ \cite{GutoskiW05} to be equal to $\cls{DQIP}(1,2)$, wherein it was shown that this class contains $\cls{QIP}=\cls{DQIP}(\mathit{poly},0)$.
The importance of short quantum games has been diminished by the proof of $\cls{QIP}=\cls{PSPACE}$, as containment of $\cls{QIP}$ is no longer such a peculiar property.
However, the containment of $\cls{PSPACE}$ inside $\cls{DQIP}(1,2)$ is still interesting, as it is not known whether $\cls{PSPACE}$ is contained in $\cls{DIP}(1,2)$, the classical version of this class.

\subsubsection{Our contribution}

As we explain in Section \ref{sec:dqip}, the oracle-algorithm of Theorem \ref{thm:main-result}---together with a parallel implementation of a suitably chosen oracle---implies that near-optimal strategies for the provers in a double quantum interactive proof can be computed efficiently in parallel.
The following containment then follows from a standard argument (summarized in Section \ref{sec:dqip:containment}).

\begin{theorem} \label{thm:dqip-in-pspace}
  $\cls{DQIP}\subseteq\cls{PSPACE}$.
\end{theorem}

This containment, when combined with the trivial containments $\cls{IP}\subseteq\cls{DIP}\subseteq\cls{DQIP}$ and the well-known fact that $\cls{PSPACE}\subseteq\cls{IP}$ \cite{LundF+92,Shamir92}, yields the following characterization.

\begin{corollary}
\label{cor:dqip-pspace}
  $\cls{DQIP}=\cls{DIP}=\cls{PSPACE}$.
\end{corollary}

As a special case of Corollary \ref{cor:dqip-pspace} we obtain the solution to an open problem of \cite{JainJ+11}:

\begin{corollary}
  \( \cls{QRG}(2)=\cls{PSPACE}. \)
\end{corollary}

Another special case of our result is a direct polynomial-space simulation of multi-message quantum interactive proofs, resulting in a first-principles proof of $\cls{QIP}=\cls{PSPACE}$.

\begin{corollary}
  \( \cls{QIP}=\cls{PSPACE} \) via direct polynomial-space simulation of multi-message quantum interactive proofs.
\end{corollary}

By contrast, all other known proofs \cite{JainJ+11,Wu10} rely upon the fact that the verifier can be assumed to exchange only three messages with the prover \cite{KitaevW00}.
The original proof of Jain \emph{et al.}\ \cite{JainJ+11} also relies on the additional fact that the verifier's only message to the prover can be just a single classical coin flip \cite{MarriottW05}.

Of course, every other competing-provers complexity class whose protocol can be cast as a double interactive proof also collapses to $\cls{PSPACE}$, such as
the aforementioned class $\cls{DQIP}(1,2)$ based on short quantum games.

It follows from the collapse of $\cls{DQIP}$ and $\cls{DIP}$ to $\cls{PSPACE}$ that these classes are closed under complement and that they are robust with respect to the choice of parameters $c,s$.
(Indeed, it may be assumed that $c=1$ and $s\leq 2^{-q}$ for any desired polynomially-bounded function $q(|x|)$---see Section \ref{sec:consequences:robust}.)

Prior to the present work polynomial-space algorithms were known only for two-turn classical interactive proofs with competing provers ($\cls{RG}(2)$), for one-turn quantum interactive proofs with competing provers ($\cls{QRG}(1)$), and for single-prover quantum interactive proofs ($\cls{QIP}$).
Our result unifies and subsumes all of these algorithms.
It also demonstrates for the first time the existence of a polynomial-space algorithm for a competing-prover interaction (classical or quantum) in which one prover reacts adaptively to the other.

Finally, our results illustrate a difference in the effect of public randomness between \emph{single}-prover interactive proofs and \emph{competing}-prover interactive proofs.
Any classical interactive proof with single prover can be simulated by another \emph{public-coin} interactive proof where the verifier's messages to the prover consist entirely of uniformly random bits and the verifier uses no other randomness \cite{GoldwasserS89}.
Extending the notion of public-coin interaction to competing-prover interactions, it is easy to see that any such interaction with a public-coin verifier can be simulated by a double interactive proof.\footnote{
  \emph{Proof sketch:} As the verifiers's questions to each prover are uniformly random, they cannot depend on prior responses from the other prover and can therefore be reordered so that all messages with one prover are exchanged before any messages with the other.
}
We therefore have that the public-coin version of $\cls{RG}$ is a subset of $\cls{DIP}$, which we now know is equal to $\cls{PSPACE}$.
Thus, by contrast to the single-prover case where
$\cls{public\textrm{-}coin\textrm{-}IP}=\cls{IP}$,
in the competing-prover case we establish the following.

\begin{corollary}
  $\cls{public\textrm{-}coin\textrm{-}RG}\neq\cls{RG}$ unless $\cls{PSPACE}=\cls{EXP}$.
\end{corollary}

\subsection{Summary of techniques}
\label{sec:intro:techniques}

\subsubsection{The matrix multiplicative weights update method}

The parallel oracle-algorithm we exhibit in the proof of Theorem \ref{thm:main-result} is an example of the \emph{matrix multiplicative weights update method (MMW)} as presented in Refs.\ \cite{AroraHK05,WarmuthK06,Kale07}.
We draw upon the valuable experience of recent applications of this method to parallel algorithms for quantum complexity classes \cite{JainW08,JainU+09,JainJ+11,Wu10}.
We also make extensive use of efficient parallel algorithms for various matrix manipulation tasks, such as computing the singular value decomposition or exponential of a matrix.
The reader is referred to von zur Gathen for more detail on parallel algorithms for matrix operations \cite{vzGathen93} and to works of Jain \emph{et al.}\ for discussion of the use of these algorithms in parallel implementations of the matrix multiplicative weights update method \cite{JainU+09,JainJ+11}.

In its unaltered form, the MMW can be used to solve min-max problems over the domain of \emph{density operators}---positive semidefinite matrices $X$ with $\ptr{}{X}=1$.
We introduce a new extension to this method for min-max problems over the domain $\bA$ defined in the SDP \eqref{eq:SDP}---a domain consisting of \emph{$k$-tuples of density operators} lying within a \emph{strict subspace} of the affine space associated with $k$-tuples of density operators.
The high-level approach of our method is as follows:
\begin{enumerate}

\item
\textbf{Extend the domain from a single density matrix to a $k$-tuple of density matrices.} \\
This step is straightforward:
the MMW can be applied without complication to all $k$ density matrices at the same time.
(Equivalently, $k$ density matrices may be viewed as a single, larger, block-diagonal density matrix.)

\item
\textbf{Restrict the domain to a strict subspace of $k$-tuples of density matrices.}\\
This step is more difficult.
It is accomplished by relaxing the problem so as to allow \emph{all} $k$-tuples, with an additional \emph{penalty term} to remove incentive for the players to use inconsistent transcripts.
\item

\textbf{Round strategies in the relaxed problem to strategies in the
original protocol.} \\
For this step one must prove a ``rounding'' theorem (Theorem \ref{thm:lambda}), which establishes that near-optimal, fully admissible strategies can be obtained from near-optimal strategies in the unrestricted domain with penalty term.

\end{enumerate}

\subsubsection{Finding optimal strategies for the provers in a double quantum interactive proof}

In Section \ref{sec:dqip} we observe that the verifier in a double quantum interactive proof induces a min-max problem of the form \eqref{eq:min-max} in which elements of $\bA$ correspond to strategies for the yes-prover and elements of $\bP$ correspond to strategies for the no-prover.
Thus, the parallel oracle-algorithm of Theorem \ref{thm:main-result}---together with a parallel implementation of the oracle for optimization over $\bP$---can be used to find optimal strategies for the provers in a double quantum interactive proof.

Our implementation of this oracle
is itself a special case of the algorithm of Theorem \ref{thm:main-result}, so that the overall algorithm employs the MMW method  \emph{twice} in a two-level recursive fashion.
At the top level the MMW is used to iteratively converge toward an optimal strategy for the yes-prover; at the bottom level the MMW is used again to solve an SDP for ``best responses'' for the no-prover to a given strategy for the yes-prover.

The central challenge in using the MMW to find optimal strategies for parties in a quantum interaction is to find a representation for strategies that is amenable to the MMW method.
In Kitaev's \emph{transcript} representation \cite{Kitaev02} the actions of a prover in a double quantum interactive proof are represented by a list $X_1,\dots,X_k$ of density matrices that satisfy a special consistency condition that is captured by the definition of the feasible region $\bA$ of the SDP \eqref{eq:SDP}.
Intuitively, these density matrices correspond to ``snapshots'' of the state of the verifier's qubits at various times during the interaction.
(See Figure \ref{fig:transcript} on page \pageref{fig:transcript}.)

The key property of double quantum interactive proofs that we exploit is the ability to draw a ``temporal line'' in the interaction before which only the yes-prover acts and after which only the no-prover acts.
Given a transcript $X_1,\dots,X_k$ for the yes-prover, the actions of the no-prover can then be represented by another transcript $Y_1,\dots,Y_\ell$.
By optimizing over all such transcripts one obtains an oracle for ``best responses'' for the no-prover to a given strategy of the yes-prover as required by the MMW method.

\subsubsection{Comparison of methods for semidefinite programming}

In their proof of $\cls{QIP}=\cls{PSPACE}$, Jain \emph{et al.}\ \cite{JainJ+11} employ the MMW to solve a special SDP for quantum interactive proofs by making direct use of the primal-dual approach described in Kale's thesis \cite{Kale07}.
Subsequent parallel algorithms for positive SDPs \cite{JainY11,PengT12} and for mixed packing and covering SDPs \cite{JainY12} are matrix generalizations (also based on MMW) of existing algorithms for linear programs \cite{LubyN93,Young01}.

We do not use any of these approaches for solving SDPs.
Instead we use the MMW to solve a min-max problem as suggested by the algorithmic proof (also presented in Kale's thesis) of a min-max theorem for a simple class of zero-sum quantum games.
By introducing a penalty term for inadmissible strategies we are able to extend this algorithm to a much richer class of games beyond the one-turn games considered by Kale.
We wish to stress that our parallel algorithm for SDPs arises as a \emph{special case} of a more general min-max algorithm, whereas previous approaches for SDPs do not generalize to min-max problems in any obvious way.

\subsubsection{Comparison of proofs of QIP $=$ PSPACE}

Unlike the present paper, the original proof of $\cls{QIP}=\cls{PSPACE}$ due to Jain \emph{et al.}\ \cite{JainJ+11} does not take advantage of the transcript representation for arbitrary multi-turn strategies.
Instead, as mentioned earlier, those authors derive a special SDP by invoking several nontrivial facts about quantum interactive proofs.
Admittedly, their SDP does bear a resemblance to Kitaev's transcript conditions, but this resemblance is only superficial and their solution applies only to a very restricted subset of transcripts.
Indeed, their derivation breaks down without the assumption that the verifier sends only classical messages to the prover.

Previously one of us \cite{Wu10} presented a simplified proof of $\cls{QIP}=\cls{PSPACE}$ that, like the work of the present paper, employs Kale's algorithmic min-max theorem \cite{Kale07} instead of the primal-dual approach for SDPs that was used in the original proof by Jain \emph{et al.}\ \cite{JainJ+11}.
The $\cls{QIP}$-completeness of the quantum circuit distinguishability problem \cite{RosgenW05} means that quantum interactive proofs can be decided by approximating the diamond norm of the difference between two quantum channels.
Wu noticed that the diamond norm can be approximated in this special case by a direct application of Kale's algorithmic min-max theorem.
His result did not require the penalization method introduced in the present paper nor an attendant rounding theorem.

\subsubsection{The Bures angle}

Finally, it is noteworthy that the proof of our rounding theorem (Theorem \ref{thm:lambda}) contains an interesting and nontrivial application of the Bures angle, which is a distance measure for quantum states that is defined in terms of the more familiar fidelity function.

Properties of the trace norm, which captures the physical distinguishability of quantum states, are sufficient for most needs in quantum information.
When some property of the fidelity is also required one uses the Fuchs-van de Graaf inequalities to convert between the trace norm and fidelity \cite{FuchsvdG99}.
(These inequalities are listed in Eq.\ \eqref{eq:Fuchs} of Section \ref{sec:prelim:angle}.)

However, every such conversion incurs a quadratic slackening of relevant accuracy parameters.
Our study calls for repeated conversions, which would incur an unacceptable exponential slackening if done naively via Fuchs-van de Graaf.
Instead, we make only a \emph{single} conversion between the trace norm and the Bures angle and then repeatedly exploit the simultaneous properties of (i) the triangle inequality, (ii) contractivity under quantum channels, and (iii) preservation of subsystem fidelity.

Although conversion inequalities between the trace norm and Bures
metric are implied by Fuchs-van de Graaf, to our knowledge explicit
conversion inequalities have not yet appeared in published
literature.
The required inequalities are derived in the present paper (Proposition \ref{prop:angle-tnorm}).

\section{Preliminaries}
\label{sec:prelim}

Hereafter we must assume familiarity with standard concepts from quantum information, though we have attempted to minimize our use of quantum formalism for the benefit of a wider audience.
The reader is referred to Nielsen and Chuang \cite{NielsenC00} and to the lecture notes of Watrous \cite{Watrous11-lec} for proper introductions to the field.
This section provides a short glossary clarifying our notation and terminology in Section \ref{sec:prelim:notation} followed by a review of two rarer but nonetheless simple and fundamental concepts from quantum information: the preservation of subsystem fidelity in Section \ref{sec:prelim:fidelity} and the Bures angle in Section \ref{sec:prelim:angle}.

\subsection{Terminology and notation}
\label{sec:prelim:notation}

\begin{description}

\item[Density matrix, quantum state.]
A \emph{density matrix} or \emph{quantum state} is a positive semidefinite matrix $X$ with $\ptr{}{X}=1$.
Thus far, we have used upper-case Roman letters ($X,Y,\dots$) to denote density matrices, as well as other matrices.
But it is standard practice in quantum information to denote density matrices with lower-case Greek letters ($\rho,\xi,\dots$).
Hereafter we adopt this convention.

\item[Measurement operator.]
A \emph{measurement operator} is a positive semidefinite matrix $M$ with $\norm{M}\leq 1$.
Equivalently, it holds that $0\preceq M\preceq I$.

\item[Quantum channel.]
A \emph{channel} is a completely positive and trace-preserving linear map $\Phi:\bbM_m\to\bbM_n$ from matrices to matrices.
These maps correspond to physically realizable operations on quantum states.

\item[Adjoint, matrix inner product.]
The \emph{adjoint} $A^*$ of a matrix $A$ is simply the conjugate-transpose of $A$.
The \emph{inner product} $\inner{A}{B}$ between two $m\times n$ matrices $A,B$ is given by $\inner{A}{B}=\ptr{}{A^*B}$.
The inner product between two $k$-tuples of matrices is given by the sum
\[ \Inner{(A_1,\dots,A_k)}{(B_1,\dots,B_k)} = \sum_{i=1}^k \inner{A_i}{B_i}. \]
More generally, the adjoint $\Phi^*$ of a linear map $\Phi$ from matrices to matrices is the unique linear map with $\inner{\Phi(X)}{Y}=\inner{X}{\Phi^*(Y)}$ for all $X,Y$.
This formula extends in the obvious way to linear maps from tuples of matrices to tuples of matrices.

\item[Trace norm.]
The \emph{trace norm} $\tnorm{X}$ of a matrix $X$ is defined as the sum of the singular values of $X$.
As a measure of distance between quantum states, the trace norm is given by
\begin{equation}
  \frac{1}{2}\tnorm{\rho-\xi}=\max_{0\preceq\Pi\preceq I} \inner{\rho-\xi}{\Pi}
  \label{eq:tnorm}
\end{equation}
for all density matrices $\rho,\xi$.

\item[Fidelity.]
The \emph{fidelity} is another distance measure for quantum states given by
\[ F(\rho,\xi) = \Tnorm{\sqrt{\rho}\sqrt{\xi}} \]
for all density matrices $\rho,\xi$.

\end{description}

\subsection{Preservation of subsystem fidelity}
\label{sec:prelim:fidelity}

Consider the following property of the fidelity function, which we call the \emph{preservation of subsystem fidelity}:
if $\rho,\xi$ are states of a quantum system with fidelity $F(\rho,\xi)$ and $\rho'$ is any state of a larger system consistent with $\rho$ then it is always possible to find $\xi'$ consistent with $\xi$ such that $F(\rho',\xi')=F(\rho,\xi)$.

A formal construction of such a $\xi'$ appears in Ref.\ \cite{JainU+09}.
Since their construction consists entirely of elementary matrix operations, there is an efficient parallel algorithm that takes as input $\rho,\xi,\rho'$ and produces the desired state $\xi'$ as output.

\begin{proposition}[Preservation of subsystem fidelity {\cite[Lemma 7.2]{JainU+09}}]
\label{prop:fidelity}
Let $\rho,\xi\in\bbM_m$ and $\rho'\in\bbM_{mn}$ be density matrices with $\ptr{\bbM_n}{\rho'}=\rho$.
There exists a density matrix $\xi'\in\bbM_{mn}$ with $\ptr{\bbM_m}{\xi'}=\xi$ and $F(\rho',\xi')=F(\rho,\xi)$.
Moreover $\xi'$ can be computed efficiently in parallel given $\rho,\xi,\rho'$.
\end{proposition}

\subsection{The Bures angle}
\label{sec:prelim:angle}

The \emph{Bures angle} or simply the \emph{angle} $A(\rho,\xi)$ between quantum states $\rho,\xi$ is defined by
\[ A(\rho,\xi) \defeq \arccos F(\rho,\xi). \]
The angle is a metric on quantum states, meaning that it is nonnegative, equals zero only when $\rho=\xi$, and obeys the triangle inequality \cite{NielsenC00}.
Moreover, the angle is \emph{contractive}, so that
\[ A(\Phi(\rho),\Phi(\xi)) \leq A(\rho,\xi) \]
for any quantum channel $\Phi$.
The Fuchs-van de Graaf inequalities establish a relationship between the fidelity and trace norm \cite{FuchsvdG99}.
The inequalities are
\begin{equation}
\label{eq:Fuchs}
  1 - F(\rho,\xi) \leq \frac{1}{2}\tnorm{\rho-\xi} \leq \sqrt{1-F(\rho,\xi)^2}.
\end{equation}
These inequalities can be used to derive a relationship between $A(\rho,\xi)$ and $\tnorm{\rho-\xi}$.
For example,

\begin{proposition}[Relationship between trace norm and Bures angle]
\label{prop:angle-tnorm}
  For all density matrices $\rho,\xi$ it holds that
  \[ \frac{1}{2}\tnorm{\rho-\xi} \leq A(\rho,\xi) \leq \sqrt{\frac{\pi}{2}\tnorm{\rho-\xi}}. \]
\end{proposition}

\begin{proof}
The lower bound on $A(\rho, \xi)$ follows immediately from Fuchs-van de Graaf:
\[ \frac{1}{2}\tnorm{\rho-\xi} \leq \sqrt{1-\cos A(\rho,\xi)^2} = \sin A(\rho,\xi) \leq A(\rho,\xi) \]
where we used the identity $\sin x \leq x$ for all $x\geq 0$.

To obtain the upper bound on $A(\rho, \xi)$ we employ the identity $\cos x \leq 1 - x^2/\pi$ for $x\in[0,\pi/2]$, which can be verified using basic calculus. 
Then we have
\[ \frac{1}{2}\tnorm{\rho-\xi} \geq 1-\cos A(\rho,\xi) \geq \frac{A(\rho,\xi)^2}{\pi} \]
from which the proposition follows.
\end{proof}

\section{Rounding theorem for a relaxed min-max problem}
\label{sec:rounding}

In this section we define a new min-max expression $\mu_\varepsilon(\bA,\bP)$ that approximates the desired quantity $\lambda(\bA,\bP)$ from \eqref{eq:min-max} in the limit as $\varepsilon$ approaches zero.
This new expression is a relaxation of $\lambda(\bA,\bP)$ that is more amenable to the MMW.
We prove a ``rounding theorem'' (Theorem \ref{thm:lambda}) by which near-optimal points for $\lambda(\bA,\bP)$ are efficiently obtained from near-optimal points for $\mu_\varepsilon(\bA,\bP)$.
Our use of the Bures angle occurs in the proof of Lemma \ref{lm:angle-bound}, which is used in the proof of our rounding theorem.

Define the relaxation $\mu_\varepsilon(\bA,\bP)$ of $\lambda(\bA,\bP)$ by
\ifthenelse{\equal{\compileCC}{1}}
{
\begin{align*}
  &\mu_\varepsilon(\bA,\bP) \\ \defeq{}&
  \min_{(\rho_1,\dots,\rho_k)}\max_{\substack{P\in\bP\\(\Pi_1,\dots,\Pi_{k-1})}}
  \Inner{\rho_k}{P} + \frac{k}{\varepsilon} \sum_{i=1}^{k-1}
  \Inner{\ptr{\bbM_n}{\rho_{i+1}} - \Phi_i(\rho_i)}{\Pi_i}\\
  ={}&
  \min_{(\rho_1,\dots,\rho_k)}\max_{P\in\bP}
  \Inner{\rho_k}{P} + \frac{k}{\varepsilon} \sum_{i=1}^{k-1}
  \frac{1}{2}\Tnorm{\ptr{\bbM_n}{\rho_{i+1}} - \Phi_i(\rho_i)}
\end{align*}
}
{
\begin{align*}
  \mu_\varepsilon(\bA,\bP) &\defeq
  \min_{(\rho_1,\dots,\rho_k)}\max_{\substack{P\in\bP\\(\Pi_1,\dots,\Pi_{k-1})}}
  \Inner{\rho_k}{P} + \frac{k}{\varepsilon} \sum_{i=1}^{k-1}
  \Inner{\ptr{\bbM_n}{\rho_{i+1}} - \Phi_i(\rho_i)}{\Pi_i}\\
  &=
  \min_{(\rho_1,\dots,\rho_k)}\max_{P\in\bP}
  \Inner{\rho_k}{P} + \frac{k}{\varepsilon} \sum_{i=1}^{k-1}
  \frac{1}{2}\Tnorm{\ptr{\bbM_n}{\rho_{i+1}} - \Phi_i(\rho_i)}
\end{align*}
}
Here the minimum is taken over all density operators $\rho_1,\dots,\rho_k\in\bbM_{mn}$ and the maximum over all $P\in\bP$ and over all measurement operators $\Pi_1,\dots,\Pi_{k-1}\in\bbM_m$.
The second equality follows immediately from the identity \eqref{eq:tnorm} from Section \ref{sec:prelim:notation}.

Notice that the minimum in the definition of $\mu_\varepsilon(\bA,\bP)$ is taken over \emph{all} $k$-tuples $(\rho_1,\dots,\rho_k)$ of density operators, not just those in $\bA$.
Each term in the summation serves to penalize any violation of the conditions required for membership in $\bA$ by adding the magnitude of that violation to the objective function.
The $k/\varepsilon$ factor amplifies the penalty so as to remove incentive to select an element outside of $\bA$.
Indeed, it is clear that
\[ \lim_{\varepsilon\to 0}\mu_\varepsilon(\bA,\bP) = \lambda(\bA,\bP). \]
The following ``rounding'' theorem establishes a specific rate of convergence for this limit.
A subsequent extension of this theorem (Proposition \ref{cor:lambda}) provides a means by which near-optimal points for $\lambda(\bA,\bP)$ are efficiently computed from near-optimal points for $\mu_\varepsilon(\bA,\bP)$.

\begin{theorem}[Rounding theorem]
\label{thm:lambda}
  For any $\varepsilon>0$ it holds that
  \( \lambda(\bA,\bP) \geq \mu_\varepsilon(\bA,\bP) > \lambda(\bA,\bP)-\varepsilon. \)
\end{theorem}

\begin{proof}

The first inequality is easy:
let $(\rho_1,\dots,\rho_k)$ be optimal for $\lambda(\bA,\bP)$ and let $(P,\Pi_1,\dots,\Pi_{k-1})$ be optimal for $\mu_\varepsilon(\bA,\bP)$.
Then we have
\ifthenelse{\equal{\compileCC}{1}}
{
\begin{align*}
  \lambda(\bA,\bP) \geq
  \Inner{\rho_k}{P} &=
  \Inner{\rho_k}{P} + \frac{k}{\varepsilon} \sum_{i=1}^{k-1}
  \Inner{\ptr{\bbM_n}{\rho_{i+1}} - \Phi_i(\rho_i)}{\Pi_i} \\ &\geq
  \mu_\varepsilon(\bA,\bP).
\end{align*}
}
{
\[
  \lambda(\bA,\bP) \geq
  \Inner{\rho_k}{P} =
  \Inner{\rho_k}{P} + \frac{k}{\varepsilon} \sum_{i=1}^{k-1}
  \Inner{\ptr{\bbM_n}{\rho_{i+1}} - \Phi_i(\rho_i)}{\Pi_i} \geq
  \mu_\varepsilon(\bA,\bP).
\]
}
(The first inequality is because $(\rho_1,\dots,\rho_k)$ is optimal for $\lambda(\bA,\bP)$.
The equality follows because $(\rho_1,\dots,\rho_k)\in\bA$, so each term in the sum is zero.
The final inequality is because $(P,\Pi_1,\dots,\Pi_{k-1})$ is optimal for $\mu_\varepsilon(\bA,\bP)$.)

The second inequality is more difficult.
We invoke the following lemma, the proof of which appears later in this section.

\begin{lemma}[Rounding lemma]
\label{cor:angle-bound}
  For any $\varepsilon>0$ and any states $\rho_1,\dots,\rho_k\in\bbM_{mn}$ there exists $(\rho_1',\dots,\rho_k')\in\bA$ such that
  \[
    \frac{1}{2}\tnorm{\rho_k-\rho_k'}
    < \varepsilon + \frac{k}{\varepsilon} \sum_{i=1}^{k-1} \frac{1}{2}
    \tnorm{\ptr{\bbM_n}{\rho_{i+1}}-\Phi_i(\rho_i)}.
  \]
  Moreover, $\rho_1',\dots,\rho_k'$ can be computed efficiently in parallel given $\rho_1,\dots,\rho_k$.
\end{lemma}

Let $(\rho_1,\dots,\rho_k)$ be optimal for $\mu_\varepsilon(\bA,\bP)$, let $(\rho_1',\dots,\rho_k')$ be the density operators obtained by invoking Lemma \ref{cor:angle-bound}, and let $P\in\bP$ be optimal for $\lambda(\bA,\bP)$.
Because $(\rho_1,\dots,\rho_k)$ is optimal for $\mu_\varepsilon(\bA,\bP)$ we have
\begin{align}
  \label{eq:cancel}
  \mu_\varepsilon(\bA,\bP) \geq
  \Inner{\rho_k}{P} + \frac{k}{\varepsilon} \sum_{i=1}^{k-1}
  \frac{1}{2}\Tnorm{\ptr{\bbM_n}{\rho_{i+1}}-\Phi_i(\rho_i)}
\end{align}
Employing the identity \eqref{eq:tnorm}, the quantity $\Inner{\rho_k}{P}$ becomes
\[
  \Inner{\rho_k}{P} = \Inner{\rho_k'}{P} + \Inner{\rho_k-\rho_k'}{P} \geq \Inner{\rho_k'}{P} - \frac{1}{2}\Tnorm{\rho_k-\rho_k'}.
\]
Substituting the bound on $\frac{1}{2}\Tnorm{\rho_k-\rho_k'}$ from Lemma \ref{cor:angle-bound}, we see that the summation of trace norms in \eqref{eq:cancel} is canceled, leaving
\[
  \mu_\varepsilon(\bA,\bP) > \Inner{\rho_k'}{P} -\varepsilon \geq \lambda(\bA,\bP) - \varepsilon
\]
as desired.
(The final inequality is because $P$ is optimal for $\lambda(\bA,\bP)$.)
\end{proof}

\begin{proposition}[Construction of near-optimal strategies]
\label{cor:lambda}
  The following hold for
  any $\delta,\varepsilon>0$:
  \begin{enumerate}

  \item \label{item:lambda:rho}

  If $(\rho_1,\dots,\rho_k)$ is $\delta$-optimal for $\mu_\varepsilon(\bA,\bP)$ then there is an efficient parallel algorithm to compute $(\rho_1',\dots,\rho_k')\in\bA$ that is $(\delta+\varepsilon)$-optimal for $\lambda(\bA,\bP)$.

  \item \label{item:lambda:P}

  If $(P,\Pi_1,\dots,\Pi_{k-1})$ is $\delta$-optimal for $\mu_\varepsilon(\bA,\bP)$ then $P$ is also $(\delta+\varepsilon)$-optimal for $\lambda(\bA,\bP)$.

  \end{enumerate}
\end{proposition}

\begin{proof}[Proof of item \ref{item:lambda:rho}]

Let $(\rho_1,\dots,\rho_k)$ be $\delta$-optimal for $\mu_\varepsilon(\bA,\bP)$, let $(\rho_1',\dots,\rho_k')\in\bA$ be obtained by invoking Lemma \ref{cor:angle-bound}, and let $P\in\bP$.
We have
\begin{align*}
  \Inner{\rho_k'}{P}
  &\leq \Inner{\rho_k}{P} + \frac{1}{2}\Tnorm{\rho_k-\rho_k'} \\
  &\leq \Inner{\rho_k}{P} + \varepsilon + \frac{k}{\varepsilon} \sum_{i=1}^{k-1} \frac{1}{2}
    \tnorm{\ptr{\bbM_n}{\rho_{i+1}}-\Phi_i(\rho_i)}\\
  &\leq \mu_\varepsilon(\bA,\bP) + \varepsilon + \delta
  \leq \lambda(\bA,\bP) + \varepsilon + \delta
\end{align*}
(The first inequality follows from \eqref{eq:tnorm};
the second from Lemma \ref{cor:angle-bound};
the third because $(\rho_1,\dots,\rho_k)$ is $\delta$-optimal for $\mu_\varepsilon(\bA,\bP)$;
and the fourth because $\mu_\varepsilon(\bA,\bP)\leq\lambda(\bA,\bP)$.)
It therefore follows that $(\rho_1',\dots,\rho_k')$ is $(\delta+\varepsilon)$-optimal for $\lambda(\bA,\bP)$.
\end{proof}

\begin{proof}[Proof of item \ref{item:lambda:P}]

Let $(P,\Pi_1,\dots,\Pi_{k-1})$ be $\delta$-optimal for $\mu_\varepsilon(\bA,\bP)$.
For any $(\rho_1,\dots,\rho_k)\in\bA$ we have
\begin{align*}
  \Inner{\rho_k}{P} &= \Inner{\rho_k}{P} + \frac{k}{\varepsilon} \sum_{i=1}^{k-1}
  \Inner{\ptr{\bbM_n}{\rho_{i+1}} - \Phi_i(\rho_i)}{\Pi_i} \\
  &\geq \mu_\varepsilon(\bA,\bP)-\delta > \lambda(\bA,\bP) - \varepsilon - \delta
\end{align*}
(The equality is because $(\rho_1,\dots,\rho_k)\in\bA$ so each term in the sum is zero.
The first inequality is because $(P,\Pi_1,\dots,\Pi_{k-1})$ is $\delta$-optimal for $\mu_\varepsilon(\bA,\bP)$.
The final inequality is because $\mu_\varepsilon(\bA,\bP)>\lambda(\bA,\bP)-\varepsilon$.)
It therefore follows that $P$ is $(\delta+\varepsilon)$-optimal for $\lambda(\bA,\bP)$.
\end{proof}

We now prove Lemma \ref{cor:angle-bound}, the statement of which appeared in the proof of Theorem \ref{thm:lambda}.
Given any states $\rho_1,\dots,\rho_k$ this lemma asserts that these states can be ``rounded'' to an element $(\rho_1',\dots,\rho_k')\in\bA$ in such a way that the distance between the final states $\rho_k$ and $\rho_k'$ is bounded by a function of the extent to which $(\rho_1,\dots,\rho_k)$ violate the conditions required for membership in $\bA$.
Let us re-state Lemma \ref{cor:angle-bound} in terms of the Bures angle.

\begin{lemma}[Rounding lemma]
\label{lm:angle-bound}
  For any $\varepsilon>0$ and any states $\rho_1,\dots,\rho_k\in\bbM_{mn}$ there exists $(\rho_1',\dots,\rho_k')\in\bA$ such that
  \[
    A(\rho_k,\rho_k') \leq \sum_{i=1}^{k-1}
    A\Pa{\ptr{\bbM_n}{\rho_{i+1}},\Phi_i(\rho_i)}.
  \]
  Moreover, $\rho_1',\dots,\rho_k'$ can be computed efficiently in parallel given $\rho_1,\dots,\rho_k$.
\end{lemma}

\begin{proof}

Define $\rho_1',\dots,\rho_k'$ recursively as follows.
Let $\rho_1'=\rho_1$.
For each $i=1,\dots,k-1$ by the preservation of subsystem fidelity (Proposition \ref{prop:fidelity}) there exists $\rho_{i+1}'$ (which can be computed efficiently in parallel) with
\( \ptr{\bbM_n}{\rho_{i+1}'}=\Phi_i(\rho_i') \)
and
\[ A(\rho_{i+1},\rho_{i+1}') = A\Pa{\ptr{\bbM_n}{\rho_{i+1}},\Phi_i(\rho_i')}. \]
By the triangle inequality this quantity is at most
\begin{align*}
  A\Pa{\ptr{\bbM_n}{\rho_{i+1}},\Phi_i(\rho_i)}
  + A\Pa{\Phi_i(\rho_i),\Phi_i(\rho_i')}.
\end{align*}
By contractivity of the Bures angle under channels, the summand on the right is at most
$A(\rho_i,\rho_i')$.
The lemma now follows inductively from the fact that $A(\rho_1,\rho_1')=0$.
\end{proof}

It is easy to recover Lemma \ref{cor:angle-bound} from Lemma \ref{lm:angle-bound}:
it follows immediately from Lemma \ref{lm:angle-bound} and Proposition \ref{prop:angle-tnorm} (Relationship between trace norm and Bures angle) that
\[
  \frac{1}{2}\tnorm{\rho_k-\rho_k'}
  \leq \sum_{i=1}^{k-1} \sqrt{\frac{\pi}{2}\Tnorm{\ptr{\bbM_n}{\rho_{i+1}}-\Phi_i(\rho_i)}}.
\]
Lemma \ref{cor:angle-bound} then follows from the fact that
$\sqrt{\frac{\pi}{2}x}<\frac{1}{2\delta}x + \delta$
for all $x\geq 0$ and all $\delta>0$.

\section{A parallel oracle-algorithm for a min-max problem}
\label{sec:alg}

In this section we prove Theorem \ref{thm:main-result} (Main result) by exhibiting an efficient parallel oracle-algorithm based on MMW for finding approximate solutions to the min-max problem \eqref{eq:min-max}.
The precise formulation of the MMW method used in this paper is stated below as Theorem \ref{thm:mwum}.
Our statement of this theorem is somewhat nonstandard: the result is usually presented in the form of an algorithm, whereas our presentation is purely mathematical.
However, a cursory examination of the literature---say, Kale's thesis \cite[Chapter 3]{Kale07}---reveals that our mathematical formulation is equivalent to the more conventional algorithmic form.

\begin{theorem}[Multiplicative weights update method {\cite[Theorem 10]{Kale07}}]
\label{thm:mwum}
Fix $\gamma\in(0,1/2)$ and $\alpha>0$.
Let $M^{(1)},\dots,M^{(T)}$ be arbitrary $d\times d$ ``loss'' matrices with $0\preceq M^{(t)}\preceq \alpha I$.
Let $W^{(1)},\dots,W^{(T)}$ be $d\times d$ ``weight'' matrices given by
\begin{align*}
  W^{(1)} &= I &
  W^{(t+1)} &= \exp\Pa{-\gamma \Pa{M^{(1)} + \cdots + M^{(t)}} }.
\end{align*}
Let $\rho^{(1)},\dots,\rho^{(T)}$ be density operators obtained by normalizing each $W^{(1)},\dots,W^{(T)}$ so that $\rho^{(t)}=W^{(t)}/\ptr{}{W^{(t)}}$.
For all density operators $\rho$ it holds that
\[
  \frac{1}{T}\sum_{t=1}^T \Inner{\rho^{(t)}}{M^{(t)}} \leq
  \Inner{\rho}{\frac{1}{T}\sum_{t=1}^T M^{(t)}} + \alpha\Pa{\gamma + \frac{\ln d}{\gamma T}}.
\]
\end{theorem}

Note that Theorem \ref{thm:mwum} holds for \emph{all} choices of loss matrices $M^{(1)},\dots,M^{(T)}$, including those for which each $M^{(t)}$ is chosen adversarially based upon $W^{(1)},\dots,W^{(t)}$.
This adaptive selection of loss matrices is typical in implementations of the MMW.

Let us establish some notation before stating our algorithm.
Let $\varepsilon>0$ and consider the linear mapping $f_{\bA,\varepsilon}$ with the property that
\ifthenelse{\equal{\compileCC}{1}}
{
\begin{align*}
  & \Inner{f_{\bA,\varepsilon}(\rho_1,\dots,\rho_k)}{(P,\Pi_1,\dots,\Pi_{k-1})} \\
  ={}& \Inner{\rho_k}{P} + \frac{k}{\varepsilon} \sum_{i=1}^{k-1}
  \Inner{\ptr{\bbM_n}{\rho_{i+1}} - \Phi_i(\rho_i)}{\Pi_i}
\end{align*}
}
{
\[
  \Inner{f_{\bA,\varepsilon}(\rho_1,\dots,\rho_k)}{(P,\Pi_1,\dots,\Pi_{k-1})}
  = \Inner{\rho_k}{P} + \frac{k}{\varepsilon} \sum_{i=1}^{k-1}
  \Inner{\ptr{\bbM_n}{\rho_{i+1}} - \Phi_i(\rho_i)}{\Pi_i}
\]
}
so that we may write
\[
  \mu_\varepsilon(\bA,\bP)
  = \min_{(\rho_1,\dots,\rho_k)}\max_{\substack{P\in\bP\\(\Pi_1,\dots,\Pi_{k-1})}}
  \Inner{f_{\bA,\varepsilon}(\rho_1,\dots,\rho_k)}{(P,\Pi_1,\dots,\Pi_{k-1})}.
\]
It is clear that the mapping $f_{\bA,\varepsilon}$ is given by
\ifthenelse{\equal{\compileCC}{1}}
{
\begin{align*}
  &f_{\bA,\varepsilon} : (\rho_1,\dots,\rho_k) \\ \mapsto{}& \Pa{
    \rho_k,
    \frac{k}{\varepsilon}\Br{ \ptr{\bbM_n}{\rho_2} - \Phi_1(\rho_1) },
    \dots,
    \frac{k}{\varepsilon}\Br{ \ptr{\bbM_n}{\rho_k} - \Phi_{k-1}(\rho_{k-1}) }
  }
\end{align*}
}
{
\[
  f_{\bA,\varepsilon} : (\rho_1,\dots,\rho_k) \mapsto \Pa{
    \rho_k,
    \frac{k}{\varepsilon}\Br{ \ptr{\bbM_n}{\rho_2} - \Phi_1(\rho_1) },
    \dots,
    \frac{k}{\varepsilon}\Br{ \ptr{\bbM_n}{\rho_k} - \Phi_{k-1}(\rho_{k-1}) }
  }
\]
}
It is tedious but straightforward to verify that the adjoint mapping $f_{\bA,\varepsilon}^*$ is given by
\[f_{\bA,\varepsilon}^* = \Pa{f_{\bA,\varepsilon,1}^*,\dots,f_{\bA,\varepsilon,k}^*} \]
where
\ifthenelse{\equal{\compileCC}{1}}
{
\begin{align*}
  f_{\bA,\varepsilon,1}^* &: (P,\Pi_1,\dots,\Pi_{k-1}) \mapsto
    -\frac{k}{\varepsilon} \Phi_1^*(\Pi_1) \\
  f_{\bA,\varepsilon,i}^* &: (P,\Pi_1,\dots,\Pi_{k-1}) \mapsto
    \frac{k}{\varepsilon} \Br{ \Pi_{i-1}\ot I - \Phi_i^*(\Pi_i) }\\
  &\textrm{for $i=2,\dots,k-1$}\\
  f_{\bA,\varepsilon,k}^* &: (P,\Pi_1,\dots,\Pi_{k-1}) \mapsto
    P + \frac{k}{\varepsilon}\Pi_{k-1}\ot I
\end{align*}
}
{
\begin{align*}
  f_{\bA,\varepsilon,1}^* &: (P,\Pi_1,\dots,\Pi_{k-1}) \mapsto
    -\frac{k}{\varepsilon} \Phi_1^*(\Pi_1) \\
  f_{\bA,\varepsilon,i}^* &: (P,\Pi_1,\dots,\Pi_{k-1}) \mapsto
    \frac{k}{\varepsilon} \Br{ \Pi_{i-1}\ot I - \Phi_i^*(\Pi_i) } \qquad \textrm{for $i=2,\dots,k-1$} \\
  f_{\bA,\varepsilon,k}^* &: (P,\Pi_1,\dots,\Pi_{k-1}) \mapsto
    P + \frac{k}{\varepsilon}\Pi_{k-1}\ot I
\end{align*}
}
Note that for any $(P,\Pi_1,\dots,\Pi_{k-1})$ it holds that
\ifthenelse{\equal{\compileCC}{1}}
{
\begin{gather}
\label{eq:f-bounds}
\begin{aligned}
  -\frac{k}{\varepsilon}I &\preceq f_{\bA,\varepsilon,1}^*(P,\Pi_1,\dots,\Pi_{k-1}) \preceq 0 \\
  -\frac{k}{\varepsilon}I &\preceq f_{\bA,\varepsilon,i}^*(P,\Pi_1,\dots,\Pi_{k-1}) \preceq \frac{k}{\varepsilon}I \qquad (2\leq i \leq k-1) \\
  0 &\preceq f_{\bA,\varepsilon,k}^*(P,\Pi_1,\dots,\Pi_{k-1}) \preceq \Pa{1+\frac{k}{\varepsilon}}I \preceq \frac{2k}{\varepsilon}I
\end{aligned}
\end{gather}
}
{
\begin{gather}
\label{eq:f-bounds}
\begin{aligned}
  -\frac{k}{\varepsilon}I &\preceq f_{\bA,\varepsilon,1}^*(P,\Pi_1,\dots,\Pi_{k-1}) \preceq 0 \\
  -\frac{k}{\varepsilon}I &\preceq f_{\bA,\varepsilon,i}^*(P,\Pi_1,\dots,\Pi_{k-1}) \preceq \frac{k}{\varepsilon}I \qquad \textrm{for $i=2,\dots,k-1$} \\
  0 &\preceq f_{\bA,\varepsilon,k}^*(P,\Pi_1,\dots,\Pi_{k-1}) \preceq \Pa{1+\frac{k}{\varepsilon}}I \preceq \frac{2k}{\varepsilon}I
\end{aligned}
\end{gather}
}
\ifthenelse{\equal{\compileCC}{0}}
{
The statement of our MMW algorithm in Figure \ref{fig:sqg-alg} employs these formulae for the adjoint.
\begin{figure}
}{
Our MMW algorithm is stated as follows. \vskip \baselineskip
}
\hrule
\begin{enumerate}

\item \label{item:gamma-T}

Let $\varepsilon=\delta/3$,
let $\gamma=\frac{\varepsilon\delta}{12k^2}$, and
let $T=\left\lceil\frac{\ln (mn)}{\gamma^2}\right\rceil$.
Let \(W_i^{(1)}=I\in\bbM_{mn}\) for each $i=1,\dots,k$.

\item

Repeat for each $t=1,\dots,T$:
\begin{enumerate}

\item

For $i=1,\dots,k$: Compute the updated density operators $\rho_i^{(t)}=W_i^{(t)}/\ptr{}{W_i^{(t)}}$.

\item \label{item:eigenspace-projection}

For $i=1,\dots,k-1$: Compute the projection $\Pi_i^{(t)}\in\bbM_m$ onto the positive eigenspace of
\[ \ptr{\bbM_n}{\rho_{i+1}^{(t)}} - \Phi_i(\rho_i^{(t)}). \]

\item

Use the oracle to obtain a $\delta/3$-optimal solution $P^{(t)}\in\bbM_{mn}$ to the optimization problem for $\bP$ (Problem \ref{problem:oracle}) on input $\rho_k^{(t)}$.

\item \label{item:loss-matrices}

Compute the loss matrices
\ifthenelse{\equal{\compileCC}{1}}
{
\begin{align*}
  & \Pa{M_1^{(t)},\dots,M_k^{(t)}} \\ ={}&
  \frac{\varepsilon}{2k^2} \Br{
    f_{R,\varepsilon}^*\Pa{P^{(t)},\Pi_1^{(t)},\dots,\Pi_{k-1}^{(t)}} +
    \frac{k}{\varepsilon}\Pa{I,\dots,I,0}.
  }
\end{align*}
}
{
\[
  \Pa{M_1^{(t)},\dots,M_k^{(t)}} =
  \frac{\varepsilon}{2k^2} \Br{
    f_{R,\varepsilon}^*\Pa{P^{(t)},\Pi_1^{(t)},\dots,\Pi_{k-1}^{(t)}} +
    \frac{k}{\varepsilon}\Pa{I,\dots,I,0}.
  }
\]
}

\item \label{item:matrix-exponential}

Update each weight matrix according to the standard MMW update rule:
\[ W_i^{(t+1)} = \exp\Pa{-\gamma\Pa{M_i^{(1)}+\cdots+M_i^{(t)}}}. \]
\end{enumerate}

\item

Return \[ \tilde \lambda = \frac{1}{T}\sum_{t=1}^T \Inner{f_{R,\varepsilon}\Pa{\rho_1^{(t)},\dots,\rho_a^{(t)}}}{\Pa{P^{(t)},\Pi_1^{(t)},\dots,\Pi_{k-1}^{(t)}}} \]
as the $\delta$-approximation to $\lambda(\bA,\bP)$.

\item \label{item:optimal-strategies}

Compute
\begin{align*}
  (\rho_1,\dots,\rho_k) &= \frac{1}{T} \sum_{t=1}^T (\rho_1^{(t)},\dots,\rho_k^{(t)})\\
  ( P,\Pi_1,\dots,\Pi_{k-1}) &= \frac{1}{T} \sum_{t=1}^T (P^{(t)},\Pi_1^{(t)},\dots,\Pi_{k-1}^{(t)}),
\end{align*}
the pair of which are $\frac{2}{3}\delta$-optimal for $\mu_\varepsilon(\bA,\bP)$.
Compute $(\rho_1',\dots,\rho_k')$ from $(\rho_1,\dots,\rho_k)$ as described in item \ref{item:lambda:rho} of Proposition \ref{cor:lambda}.
Return $(\rho_1',\dots,\rho_k')$ and $P$ as the $\delta$-optimal point for $\lambda(\bA,\bP)$.

\end{enumerate}
\hrule
\ifthenelse{\equal{\compileCC}{0}}
{
\caption{
  An  parallel oracle-algorithm for finding approximate solutions to $\lambda(\bA,\bP)$ (Problem \ref{problem:lambda}) used in the proof of Theorem \ref{thm:main-result}.
}
\label{fig:sqg-alg}
\end{figure}
}{\vskip\baselineskip}
We are now ready to prove Theorem \ref{thm:main-result}.

\begin{proof}[Proof of Theorem \ref{thm:main-result}]

\ifthenelse{\equal{\compileCC}{1}}
{
We argue that the theorem is established by the above oracle-algorithm.
}
{
We argue that the theorem is established by the oracle-algorithm presented in Figure \ref{fig:sqg-alg}.
}
To this end, note that each loss matrix $M_i^{(t)}\in\bbM_{mn}$ satisfies $0\preceq M_i^{(t)} \preceq \frac{1}{k}I$---a fact that follows immediately from their definition in step \ref{item:loss-matrices} and the bounds \eqref{eq:f-bounds} on the adjoint mapping $f_{\bA,\varepsilon}^*$.

For each $i=1,\dots,k$ it is clear that the construction of the density operators $\rho_i^{(t)}$ in terms of the loss matrices $M_i^{(t)}$ presented in
\ifthenelse{\equal{\compileCC}{1}}
{the above oracle-algorithm}
{Figure \ref{fig:sqg-alg}}
are as defined in Theorem \ref{thm:mwum}.
It therefore follows that for any density operator $\rho_i^\star\in\bbM_{mn}$ we have
\[
  \frac{1}{T}\sum_{t=1}^T \Inner{ \rho_i^{(t)} }{ M_i^{(t)} } \leq
  \Inner{ \rho_i^\star }{ \frac{1}{T}\sum_{t=1}^T M_i^{(t)} } + \frac{1}{k}\Pa{\gamma + \frac{\ln (mn)}{\gamma T}}.
\]
Summing these inequalities over all $i$ we find that for any density operators $(\rho_1^\star,\dots,\rho_k^\star)$ it holds that
\begin{align*}
  & \frac{1}{T}\sum_{t=1}^T \Inner{ \Pa{ \rho_1^{(t)},\dots,\rho_k^{(t)} } }{ \Pa{ M_1^{(t)},\dots,M_k^{(t)} } } \\ \leq{}&
  \Inner{ (\rho_1^\star,\dots,\rho_k^\star) }{ \frac{1}{T}\sum_{t=1}^T \Pa{M_1^{(t)},\dots,M_k^{(t)} } } + \Pa{\gamma + \frac{\ln (mn)}{\gamma T}}.
\end{align*}
Substituting the definition of the loss matrices $M_i^{(t)}$ from step \ref{item:loss-matrices} and simplifying, we obtain
\begin{gather}
\label{eq:tilde-lambda}
\begin{aligned}
  \tilde \lambda &=
    \frac{1}{T}\sum_{t=1}^T \Inner{ \Pa{ \rho_1^{(t)},\dots,\rho_k^{(t)} } }{ f_{R,\varepsilon}^*\Pa{ P^{(t)},\Pi_1^{(t)},\dots,\Pi_{k-1}^{(t)} } }  \\
  &\leq
  \Inner{ (\rho_1^\star,\dots,\rho_k^\star) }{ \frac{1}{T}\sum_{t=1}^T f_{R,\varepsilon}^* \Pa{ P^{(t)},\Pi_1^{(t)},\dots,\Pi_{k-1}^{(t)} } } + \underbrace{ \frac{2k^2}{\varepsilon}\Pa{\gamma + \frac{\ln (mn)}{\gamma T}} }_{\textrm{error term}}.
\end{aligned}
\end{gather}
Substituting the choice of $\gamma,T$ from step \ref{item:gamma-T} we see that the error term on the right side is at most $\delta/3$.
Since this inequality holds for any choice of $(\rho_1^\star,\dots,\rho_k^\star)$ it certainly holds for the optimal choice, from which it follows that the right side is at most $\mu_\varepsilon(\bA,\bP)+\delta/3$.
By construction each $\pa{P^{(t)},\Pi_1^{(t)},\dots,\Pi_{k-1}^{(t)}}$ is a $\delta/3$-best response to $\pa{\rho_1^{(t)},\dots,\rho_k^{(t)}}$ so it must be that the left side of this inequality is at least $\mu_\varepsilon(\bA,\bP) - \delta/3$.
It then follows from Theorem \ref{thm:lambda} (Rounding theorem) and the choice $\varepsilon=\delta/3$ that $\abs{\tilde\lambda - \lambda(\bA,\bP))}<\frac{2}{3}\delta$ as desired.

Next we argue that the point $(\rho_1',\dots,\rho_k')$ returned in step \ref{item:optimal-strategies} is $\delta$-optimal for $\lambda(\bA,\bP)$.
By item \ref{item:lambda:rho} of Proposition \ref{cor:lambda} it suffices to argue that $(\rho_1,\dots,\rho_k)$ is $\frac{2}{3}\delta$-optimal for $\mu_\varepsilon(\bA,\bP)$.
To this end, choose any $(P^\star,\Pi_1^\star,\dots,\Pi_a^\star)$.
Since each $\pa{P^{(t)},\Pi_1^{(t)},\dots,\Pi_{k-1}^{(t)}}$ is a $\delta/3$-best response to $\pa{\rho_1^{(t)},\dots,\rho_k^{(t)}}$ it holds that the inner product
\[
  \Inner{ \Pa{ \rho_1^{(t)},\dots,\rho_k^{(t)} } }{ f_{R,\varepsilon}^*\Pa{ P^{(t)},\Pi_1^{(t)},\dots,\Pi_{k-1}^{(t)} } }
\]
can increase by no more than $\delta/3$ when $(P^\star,\Pi_1^\star,\dots,\Pi_{k-1}^\star)$ is substituted for $\pa{P^{(t)},\Pi_1^{(t)},\dots,\Pi_{k-1}^{(t)}}$.
It then follows from \eqref{eq:tilde-lambda} that
\ifthenelse{\equal{\compileCC}{1}}
{
\begin{align*}
  &\Inner{ \frac{1}{T} \sum_{t=1}^T \Pa{ \rho_1^{(t)},\dots,\rho_k^{(t)} } }{ f_{R,\varepsilon}^*\Pa{ P^\star,\Pi_1^\star,\dots,\Pi_{k-1}^\star } } \\
  \leq{}& \tilde\lambda + \delta/3 \leq \mu_\varepsilon(\bA,\bP) + {\textstyle\frac{2}{3}}\delta
\end{align*}
}
{
\[
  \Inner{ \frac{1}{T} \sum_{t=1}^T \Pa{ \rho_1^{(t)},\dots,\rho_k^{(t)} } }{ f_{R,\varepsilon}^*\Pa{ P^\star,\Pi_1^\star,\dots,\Pi_{k-1}^\star } } \leq \tilde\lambda + \delta/3 \leq \mu_\varepsilon(\bA,\bP) + {\textstyle\frac{2}{3}}\delta
\]
}
and hence $(\rho_1,\dots,\rho_k)$ is $\frac{2}{3}\delta$-optimal for $\mu_\varepsilon(\bA,\bP)$ as desired.

Next we argue that the operator $P$ returned in step \ref{item:optimal-strategies} is $\delta$-optimal for $\lambda(\bA,\bP)$.
By item \ref{item:lambda:P} of Proposition \ref{cor:lambda} it suffices to argue that $(P,\Pi_1,\dots,\Pi_{k-1})$ is $\frac{2}{3}\delta$-optimal for $\mu_\varepsilon(\bA,\bP)$.
To this end, choose any $(\rho_1^\star,\dots,\rho_k^\star)$.
It follows from \eqref{eq:tilde-lambda} that
\[
  \Inner{ (\rho_1^\star,\dots,\rho_k^\star) }{ f_{R,\varepsilon}^*\Pa{ P,\Pi_1,\dots,\Pi_{k-1} } } \geq \tilde\lambda - \delta/3 \geq \mu_\varepsilon(\bA,\bP) - {\textstyle\frac{2}{3}}\delta
\]
and hence $(P,\Pi_1,\dots,\Pi_{k-1})$ is $\frac{2}{3}\delta$-optimal for $\mu_\varepsilon(\bA,\bP)$ as desired.

The efficiency of this algorithm is not difficult to argue.
Each individual step consists only of matrix operations that are known to admit an efficient parallel implementation.
Efficiency then follows from the observation that the number $T$ of iterations is polynomial in $k$, $1/\delta$, and $\log (mn)$.
\end{proof}

\section{Double quantum interactive proofs}
\label{sec:dqip}

In this section we prove $\cls{DQIP}\subseteq\cls{PSPACE}$ by means of Theorem \ref{thm:main-result}.
Specifically, in Section \ref{sec:dqip:yes} we argue that the verifier in a double quantum interactive proof induces a min-max problem of the form \eqref{eq:min-max} in which elements of $\bA$ correspond to strategies for the yes-prover, elements of $\bP$ correspond to strategies for the no-prover, and the value $\lambda(\bA,\bP)$ corresponds to the probability with which the verifier rejects when both provers act optimally.

Thus, the parallel oracle-algorithm of Theorem \ref{thm:main-result}---together with a parallel implementation of the oracle for optimization over $\bP$---can be used to compute this probability to sufficient accuracy so as to determine which prover has the winning strategy.
In Section \ref{sec:dqip:oracle} we provide a parallel implementation of the oracle
required by Theorem \ref{thm:main-result}.
Finally, in Section \ref{sec:dqip:containment} we recite the argument by which the existence of a parallel algorithm for approximating $\lambda(\bA,\bP)$ leads to the containment of $\cls{DQIP}$ inside $\cls{PSPACE}$.
First, we briefly introduce new notation in Section \ref{sec:dqip:prelim}.

\subsection{Notation}
\label{sec:dqip:prelim}

Until now we have used the symbol $\bbM_n$ to denote the space of complex $n\times n$ matrices.
This notation is ideal when only one or two distinct quantum systems are under consideration.
However, discussion henceforth deals with many different systems (called \emph{registers}) and so we adopt the convention that distinct finite-dimensional complex vector spaces of the form $\mathbb{C}^d$ shall be denoted with calligraphic letters ($\cX,\cY,\dots$).
We also adopt the following notation:
\begin{center}
\begin{tabularx}{\textwidth}{lX}
$\cX\cY$ & Shorthand for the Kronecker product $\cX\ot\cY$.
  If $\cX=\mathbb{C}^d$ and $\cY=\mathbb{C}^{d'}$ then $\cX\cY=\mathbb{C}^{dd'}$.\\
$\bbM_\cX$ & The complex space of all linear operators (matrices) acting on $\cX$.\\
$I_\cX\in\bbM_\cX$ & The identity operator acting on $\cX$.\\
$\trace_\cX:\bbM_{\cX\cY}\to\bbM_\cY$ & The partial trace over $\cX$.
\end{tabularx}
\end{center}

\subsection{Characterization of strategies for the yes-prover}
\label{sec:dqip:yes}

The verifier in a double quantum interactive proof can be assumed to act upon two quantum registers: an $m$-qubit register $\sM$ that is shared with the provers for the purpose of exchanging messages and a $v$-qubit register $\sV$ that serves as a private memory for the verifier.
Associated with the registers $\sM,\sV$ are complex Euclidean spaces $\cM=\mathbb{C}^{2^m},\cV=\mathbb{C}^{2^v}$, respectively.
A verifier who exchanges $a$ rounds of messages with the yes-prover followed by $b$ rounds of messages with the no-prover is completely specified by a tuple $V=(\ket{\psi},V_1,\dots,V_{a+b-1},\Pi)$ where
\begin{enumerate}
\item
$\ket{\psi}\in\cM\cV$ is a pure state.
\item
$V_1,\dots,V_{a+b-1}\in\bbM_{\cM\cV}$ are unitary operators.
\item
$\Pi\in\bbM_{\cM\cV}$ is a projective measurement operator.
\end{enumerate}
The yes-prover acts upon the shared communication register $\sM$ and a private memory register $\sW$ with associated space $\cW$.
The actions of the yes-prover are specified by unitaries $A_1,\dots,A_a\in\bbM_{\cM\cW}$.
Similarly, the no-prover acts upon the shared communication register $\sM$ and a private memory register $\sZ$ with associated space $\cZ$.
The actions of the no-prover are specified by unitaries $B_1,\dots,B_b\in\bbM_{\cM\cZ}$.
The interaction proceeds as suggested by Figure \ref{fig:dqip} with measurement outcome $\Pi$ indicating rejection.

\begin{figure}
\begin{center}
\ifthenelse{\equal{\compileCC}{1}}
{\includegraphics[scale=0.77]{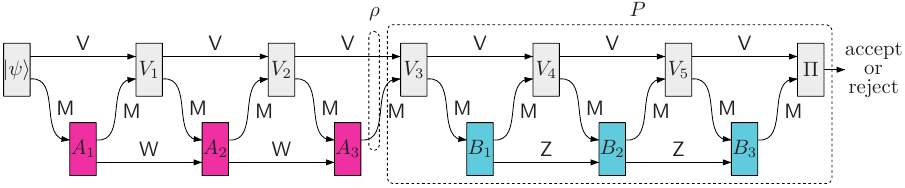}}
{\includegraphics{fig-dqip-03}}
\end{center}
  \caption{
    An illustration of a double quantum interactive proof in which the verifier $V=(\ket{\psi},V_1,\dots,V_5,\Pi)$ exchanges $a=3$ rounds of messages with the yes-prover followed by $b=3$ rounds of messages with the no-prover before performing the measurement $\set{\Pi,I-\Pi}$ that dictates acceptance or rejection.
    Any choice of $(A_1,A_2,A_3)$ and $(B_1,B_2,B_3)$ induces a state $\rho$ and a measurement operator $P$ as indicated.
    The probability of rejection is given by $\inner{\rho}{P}=\ptr{}{\rho P}$.}
  \label{fig:dqip}
\end{figure}

Basic quantum formalism tells us that if the yes- and no-provers act according to $\vec{A}=(A_1,\dots,A_a)$ and $\vec{B}=(B_1,\dots,B_b)$, respectively, then the probability of rejection is given by
\ifthenelse{\equal{\compileCC}{1}}
{
\begin{gather}
\label{eq:bob-wins}
\begin{aligned}
  & \Pr\Br{\textrm{reject} \mid \vec{A}, \vec{B} }\\
  ={}& \Norm{ \Pi B_b V_{a+b-1} B_{b-1} \cdots B_1 V_a A_a V_{a-1} A_{a-1} \cdots A_2 V_1 A_1 \ket{\psi} }^2.
\end{aligned}
\end{gather}
}
{
\begin{align}
  \Pr\Br{\textrm{reject} \mid \vec{A}, \vec{B} }
  = \Norm{ \Pi B_b V_{a+b-1} B_{b-1} \cdots B_1 V_a A_a V_{a-1} A_{a-1} \cdots A_2 V_1 A_1 \ket{\psi} }^2. \label{eq:bob-wins}
\end{align}
}
(For clarity we have suppressed numerous tensors with identity and the initial states $\ket{0}$ of the provers' private memory registers.)

For any $\vec{A}$ let $\rho$ be the reduced state of the verifier's registers $(\sM,\sV)$ immediately after $A_a$ is applied so that the actions of the yes-prover are completely represented by the state $\rho$.
Similarly, for any $\vec{B}$ let $P$ be the measurement operator on $(\sM,\sV)$ obtained by bundling the verifier--no-prover interaction into a single measurement operator as suggested by Figure \ref{fig:dqip}.
The expression \eqref{eq:bob-wins} for the probability of rejection can be rewritten in terms of $\rho,P$ as
\[ \Pr[\textrm{reject} \mid \vec{A}, \vec{B} ] = \inner{\rho}{P}. \]
By definition, the no-prover wishes to maximize this quantity while the yes-prover wishes to minimize it.
Let $\lambda(V)$ denote the verifier's probability of rejection when both provers act optimally.
For a verifier with completeness $c$ and soundness $s$, or goal is to determine whether $\lambda(V)$ is closer to $1-c$ or to $1-s$.

Let $\bY(V)\subset\bbM_{\cM\cV}$ denote the set of states of $(\sM,\sV)$ obtainable by the yes-prover and let $\bP(V)\subset\bbM_{\cM\cV}$ denote the set of measurement operators on $(\sM,\sV)$ obtainable by the no-prover.
Then the desired quantity $\lambda(V)$ is given by the min-max problem
\begin{equation}
  \label{eq:yp}
  \lambda(V) =
  \min_{\rho\in\bY(V)} \ \max_{P\in\bP(V)} \ \inner{\rho}{P}.
\end{equation}
What can be said of the sets $\bY(V),\bP(V)$?
Let us begin by considering the set $\bY(V)$.
As suggested by Figure \ref{fig:transcript}, each element of $\bY(V)$ can be viewed as the final entry $\rho_a$ in a \emph{transcript} $(\rho_1,\dots,\rho_a)$ of the verifier's conversation with the yes-prover.
\begin{figure}[t]
\begin{center}
\includegraphics{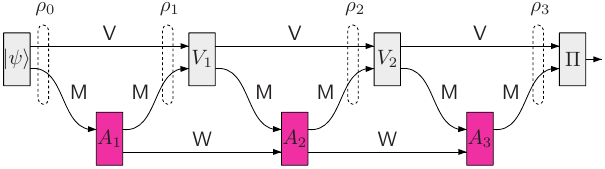}
\end{center}
  \caption{The states $\rho_1,\rho_2,\rho_3$ are a transcript of the referee's conversation with the yes-prover.
  It follows easily from the unitary equivalence of purifications that a triple $(\rho_1,\rho_2,\rho_3)$ is a valid transcript if and only if it obeys the recursive relation $\ptr{\cM_i}{\rho_i} = \ptr{\cA_i}{V_{i-1}\rho_{i-1}V_{i-1}^*}$ for $i=1,2,3$ where $V_0=I$.}
  \label{fig:transcript}
\end{figure}
Moreover, it is straightforward to use the unitary equivalence of purifications to characterize those $a$-tuples of density matrices which constitute valid transcripts.
This characterization was first noted by Kitaev \cite{Kitaev02}.

\begin{proposition}[Kitaev's consistency conditions \cite{Kitaev02}]
\label{prop:consistency}
  Let $V=(\ket{\psi},V_1,\dots,V_{a+b-1},\Pi)$ be a verifier and let $\bY(V)$ be the set of admissible states for the yes-prover.
  A given state $\rho$ is an element of $\bY(V)$ if and only if there exist density matrices $\rho_1,\dots,\rho_a\in\bbM_{\cM\cV}$ with $\rho_a=\rho$ and
  \begin{alignat*}{2}
    \ptr{\cM}{\rho_i} &= \ptr{\cM}{V_{i-1}\rho_{i-1}V_{i-1}^*} \quad & \textrm{for $i=1,\dots,a$}
  \end{alignat*}
  where we have written $V_0=I$ and $\rho_0=\ket{\psi}\bra{\psi}$ for convenience.
\end{proposition}

With these observations in mind we consider completely positive and trace-preserving linear maps
\[ \Phi_0,\dots,\Phi_{a-1}:\bbM_{\cM\cV}\to\bbM_{\cV} \]
defined by
\begin{align*}
  \Phi_0 &: X \mapsto \ptr{}{X}\ptr{\cM}{\ket{\psi}\bra{\psi}} \\
  \Phi_i &: X \mapsto \ptr{\cM}{V_i X V_i^*} \quad \textrm{ for $i=1,\dots,a-1$}
\end{align*}
These maps specify the feasible region $\bA(V)$ of an SDP of the form \eqref{eq:SDP} from Section \ref{sec:intro}.
Moreover, it follows from Kitaev's consistency conditions (Proposition \ref{prop:consistency}) that $(\rho_0,\dots,\rho_a)\in\bA(V)$ if and only if $\rho_a\in\bY(V)$.
Thus, the min-max problem \eqref{eq:yp} for $\lambda(V)$ can equivalently be written
\begin{align}
  \lambda(V) =
  \min_{(\rho_0,\dots,\rho_a)\in\bA(V)}\ \max_{P\in\bP(V)} \ \Inner{\rho_a}{P}. \label{eq:lambda-consistent}
\end{align}
We have not yet shown that the set $\bP(V)$ of measurement operators for the no-prover is compact and convex.
But if we assume for the moment that it is then we may already apply Theorem \ref{thm:main-result} so as to obtain a parallel oracle-algorithm for approximating $\lambda(V)$ on input $\Phi_0,\dots,\Phi_{a-1}$ given an oracle for optimization over $\bP(V)$.

\subsection{Implementation of the oracle for best responses of the no-prover}
\label{sec:dqip:oracle}

In order to complete the description of our parallel algorithm for double quantum interactive proofs it remains only to describe the implementation of the oracle for optimization for $\bP(V)$ (Problem \ref{problem:oracle}).
In this section we establish the following.

\begin{proposition} \label{thm:oracle}
  Let $V=(\ket{\psi},V_1,\dots,V_{a+b-1},\Pi)$ be a verifier and let $\bP(V)$ be the set of admissible measurement operators for the no-prover.
  There is a parallel algorithm for optimization over $\bP(V)$ (Problem \ref{problem:oracle}) with run time bounded by a polynomial in $b$, $1/\delta$, and $\log (\dim(\cM\cV))$.

  It follows that the algorithm of \ifthenelse{\equal{\compileCC}{1}}{Section \ref{sec:alg}}{Figure \ref{fig:sqg-alg}} yields an unconditionally efficient parallel algorithm for approximating $\lambda(V)$ given an explicit matrix representation of the verifier $V$.
\end{proposition}

As mentioned earlier, this instance of optimization over $\bP(V)$ (Problem \ref{problem:oracle}) will be rephrased as an SDP of the form \eqref{eq:SDP} (plus some post-processing) so that the algorithm of Section \ref{sec:alg} can be reused in the implementation of our oracle.

To this end choose any state $\rho\in\bbM_{\cM\cV}$ and suppose that a (possibly cheating) yes-prover was somehow able to make it so that the registers $(\sM,\sV)$ after the interaction with the yes-prover are in state $\rho$.
Let $\sW$ be a register large enough to admit a purification of $\rho$ and let $\ket{\varphi}\in\cW\cM\cV$ be any such purification.
If the no-prover acts according to $(B_1,\dots,B_b)$ then the probability of rejection (as per Eq.\ \eqref{eq:bob-wins}) is
\[
 \Pr[\textrm{reject} \mid \rho, (B_1,\dots,B_b) ] =
  \Norm{ \Pi B_b V_{a+b-1} B_{b-1} \cdots B_1 V_a \ket{\varphi} }^2.
\]
Notice that this quantity also represents the probability of rejection in a different, single-prover interactive proof with a verifier $V'$ whose initial state is $V_a\ket{\varphi}$.
(Formally, the verifier $V'$ exchanges $b$ rounds of messages with one of the provers and zero messages with the other.)
The unitaries $B_1,\dots,B_b$ could specify actions for either the yes-prover or the no-prover---a choice that depends only upon how we label the components of the verifier $V'$.

Since our goal is to reduce optimization over $\bP(V)$ (which is a maximization problem) to an SDP of the form \eqref{eq:SDP} (which is a minimization problem), it befits us to view $B_1,\dots,B_b$ as actions for the yes-prover in the interactive proof with verifier $V'$.
Let us write
\[ V'=(V_a\ket{\varphi},V_1',\dots,V_{b-1}',\Pi') \]
where $V_1',\dots,V_{b-1}',\Pi'\in\bbM_{\cM\cV\cW}$ are given by
\begin{align*}
  V_i' &= V_{a+i}\ot I_{\cW} \qquad \textrm{for $i=1,\dots,b-1$} \\
  \Pi' &= (I-\Pi)\ot I_{\cW}.
\end{align*}
The private memory register $\sV'$ of the new verifier $V'$ is identified with the registers $(\sV,\sW)$ and communication register $\sM'$ of the new verifier is identified with $\sM$.

Each choice of unitaries $(B_1,\dots,B_b)$ induces both a measurement operator $P\in\bP(V)$ and a state $\xi\in\bY(V')$ with
\[ \Inner{\rho}{P} = \Norm{ \Pi B_b V_{a+b-1} B_{b-1} \cdots B_1 V_a \ket{\varphi} }^2 = 1 - \Inner{\xi}{\Pi'} \]
and therefore
\[ \max_{P\in\bP(V)} \Inner{\rho}{P} = 1 - \lambda(V') = 1 - \min_{\xi\in\bY(V')} \Inner{\xi}{\Pi'}. \]
Moreover, $P\in\bP(V)$ achieves the maximum on the left side if and only if the unitaries $(B_1,\dots,B_b)$ that induce $P$ also induce a state $\xi\in\bY(V')$ that achieves the minimum on the right side.

Incidentally, by identifying elements of $\bP(V)$ with elements of $\bA(V')$ we have established that the set $\bP(V)$ is compact and convex as required by Theorem \ref{thm:main-result}.
We are now ready to prove Proposition \ref{thm:oracle}.

\begin{proof}[Proof of Proposition \ref{thm:oracle}]

Consider the following algorithm for optimization over $\bP(V)$:
\begin{enumerate}

\item \label{item:oracle-overview:xi}
Use the algorithm of \ifthenelse{\equal{\compileCC}{1}}{Section \ref{sec:alg}}{Figure \ref{fig:sqg-alg}} to find $\xi\in\bY(V')$ minimizing $\Inner{\xi}{\Pi'}$.

\item \label{item:oracle-overview:unitaries}
Find the unitaries $(B_1,\dots,B_b)$ that induce $\xi$.
These unitaries also induce a measurement operator $P\in\bP(V)$ maximizing $\Inner{\rho}{P}$.
Compute $P$ using $(B_1,\dots,B_b)$ via standard matrix multiplication.

\end{enumerate}
We already saw how the algorithm of \ifthenelse{\equal{\compileCC}{1}}{Section \ref{sec:alg}}{Figure \ref{fig:sqg-alg}} can be used to accomplish step \ref{item:oracle-overview:xi} given an oracle for optimization over $\bP(V')$.
In this case $\bP(V')=\set{\Pi'}$ is a singleton set and thus the oracle for optimization over $\bP(V')$ admits a trivial implementation by returning the only element.

It remains only to fill in the details for step \ref{item:oracle-overview:unitaries}.
Recall that the algorithm of \ifthenelse{\equal{\compileCC}{1}}{Section \ref{sec:alg}}{Figure \ref{fig:sqg-alg}} finds a near-optimal transcript $(\xi_0,\dots,\xi_b)\in\bA(V')$, meaning that
\begin{align*}
  \ptr{\cM}{\xi_1} &= \ptr{\cM}{V_a\ket{\varphi}\bra{\varphi}V_a^*} \\
  \ptr{\cM}{\xi_{i+1}} &= \ptr{\cM}{V_i'\xi_iV_i'^*} \qquad \textrm{for each $i=1,\dots,b-1$}.
\end{align*}
(Here $\xi_0$ is an arbitrary density matrix that is not used in our construction.
The presence of this matrix is an artifact of the identification of $\bY(V')$ with $\bA(V')$.)
The following algorithm finds the unitaries $(B_1,\dots,B_b)$:
\begin{enumerate}

\item
Let $\cZ$ be a space large enough to admit purifications of $\xi_1,\dots,\xi_b$.
Write $\ket{\alpha_0}=\ket{\varphi}\ket{0_{\cZ}}$ and $V_0'=V_a$.

\item
For each $i=1,\dots,b$:
\begin{enumerate}

\item \label{it:purify}
Compute a purification $\ket{\alpha_i}\in\cZ\cM\cV\cW$ of $\xi_i$.

\item \label{it:unitary}
Compute a unitary $B_i\in\bbM_{\cZ\cM}$ that maps $V_{i-1}'\ket{\alpha_{i-1}}$ to $\ket{\alpha_i}$.

\end{enumerate}

\item
Return the desired unitaries $(B_1,\dots,B_b)$.

\end{enumerate}
Correctness of this construction is straightforward (though notationally cumbersome).
Let us argue that each individual step consists only of matrix operations that are known to admit an efficient parallel implementation, from which it follows that the entire construction is efficient.

Step \ref{it:purify} requires that we compute a purification $\ket{\alpha}$ of a given mixed state $\xi$.
This can be achieved by computing a spectral decomposition \[\xi=\sum_i\mu_i\ket{\phi_i}\bra{\phi_i}\] of $\xi$; the purification $\ket{\alpha}$ is then given by \[\ket{\alpha}=\sum_i\sqrt{\mu_i}\ket{\phi_i}\ket{\phi_i}.\]
Given two pure states $\ket{\alpha},\ket{\alpha'}\in\cZ\cM\cV\cW$ with \[\ptr{\cZ\cM}{\ket{\alpha}\bra{\alpha}} = \ptr{\cZ\cM}{\ket{\alpha'}\bra{\alpha'}},\] step \ref{it:unitary} requires that we compute a unitary $B\in\bbM_{\cZ\cM}$ that maps $\ket{\alpha}$ to $\ket{\alpha'}$.
This can be achieved by computing Schmidt decompositions
\begin{align*}
  \ket{\alpha} &= \sum_i s_i \ket{\phi_i}\ket{\psi_i} &
  \ket{\alpha'} &= \sum_i s_i' \ket{\phi_i'}\ket{\psi_i}
\end{align*}
with respect to the partition $\cZ\cM\ot\cV\cW$.
(Schmidt decompositions on vectors are equivalent to singular value decompositions on matrices and hence can be implemented in parallel.)
The desired unitary is then given by straightforward matrix multiplication and summation:
\( B = \sum_i \ket{\phi_i'}\bra{\phi_i}. \)
\end{proof}

\subsection{Containment of DQIP inside PSPACE}
\label{sec:dqip:containment}

The argument by which a parallel algorithm for double quantum interactive proofs leads to a proof of $\cls{DQIP}\subseteq\cls{PSPACE}$ is by now a familiar one.
(See Section 3 of Ref.\ \cite{JainJ+11} for a good exposition of this type of argument.)

\begin{proof}[Proof of Theorem \ref{thm:dqip-in-pspace}]

For each decision problem $L\in\cls{DQIP}$ we must prove that there is a polynomial space algorithm for $L$.
To this end consider a ``scaled up'' version of $\cls{NC}$ known as $\cls{NC}(\mathit{poly})$, which consists of all functions computable by polynomial-space uniform Boolean circuits of polynomial depth.
It has long since been known that $\cls{NC}(\mathit{poly})$ algorithms can be simulated in polynomial space \cite{Borodin77}, so in order to prove $L\in\cls{PSPACE}$ it suffices to give an $\cls{NC}(\mathit{poly})$ algorithm for $L$.

Let $V$ be a verifier with completeness $c$, soundness $s$, and polynomial-bounded $p$ with $c-s\geq 1/p$ witnessing the membership of $L$ in $\cls{DQIP}$.
Let $x$ be any input string and consider the following algorithm for deciding whether $x$ is a yes-instance or a no-instance of $L$:
\begin{enumerate}

\item \label{it:nc-poly}
Compute an explicit matrix representation of the verifier $V=(\ket{\psi},V_1,\dots,V_{a+b-1},\Pi)$ on input $x$.
As argued earlier, this representation specifies sets $\bA(V),\bP(V)$ for a min-max problem of the form \eqref{eq:min-max}.

\item \label{it:nc}
Compute a $\delta$-approximation of $\lambda(V)$ for the choice $\delta=(c-s)/3$ so as to determine which of the two provers has a winning strategy.
Accept or reject accordingly.

\end{enumerate}
The dimension $\dim(\cM\cV)=2^{m+v}$ of the matrix representation of a verifier on input $x$ might grow exponentially in the bit length of $x$.
Nevertheless, as argued in Ref.\ \cite{JainJ+11} for ordinary quantum interactive proofs, it is not difficult to see that step \ref{it:nc-poly} admits a straightforward implementation in $\cls{NC}(\mathit{poly})$ via standard matrix multiplication.

Earlier in this section we argued that the parallel oracle-algorithm of Theorem \ref{thm:main-result} can be used to compute the desired approximation of $\lambda(V)$.
We also presented a parallel implementation of the oracle for optimization over $\bP(V)$ required by Theorem \ref{thm:main-result}.
To see that this parallel algorithm is \emph{efficient} it suffices to observe that the number of rounds $a+b$ and the inverse of the accuracy parameter $1/\delta$ both scale as a polynomial in $|x|$ and hence also in $\log(\dim(\cM\cV))$.

Thus, the above algorithm computes the composition of a function in $\cls{NC}(\mathit{poly})$ with another function in $\cls{NC}$.
As $\cls{NC}(\mathit{poly})$ is closed under such compositions, it follows that the above algorithm admits an $\cls{NC}(\mathit{poly})$ implementation and hence also a polynomial-space implementation.
It follows that $L\in\cls{PSPACE}$ and hence $\cls{DQIP}\subseteq\cls{PSPACE}$.
\end{proof}

\section{Consequences and extensions}
\label{sec:consequences}

\subsection{A direct polynomial-space simulation of QIP}

As mentioned in the introduction, a special case of our result is a direct polynomial-space simulation of multi-message quantum interactive proofs, resulting in a first-principles proof of $\cls{QIP}\subseteq\cls{PSPACE}$.
Recall that an ordinary, single-prover quantum interactive proof is a double quantum interactive proof in which the verifier exchanges zero messages with the no-prover.
We already observed in Section \ref{sec:dqip:oracle} that such a verifier induces an SDP of the form \eqref{eq:SDP} in which elements of the feasible region $\bA$ are identified with strategies for the prover.
In this case, Theorem \ref{thm:main-result} yields an efficient parallel algorithm for finding optimal strategies for the prover in a single-prover quantum interactive proof with no need to specify an oracle.

\subsection{Finding near-optimal strategies}

The algorithm of \ifthenelse{\equal{\compileCC}{1}}{Section \ref{sec:alg}}{Figure \ref{fig:sqg-alg}} not only approximates the value $\lambda(\bA,\bP)$ of the min-max problem \eqref{eq:min-max}, but it also finds near-optimal points $(\rho_1,\dots,\rho_k)\in\bA$ and $P\in\bP$.
By contrast, in Section \ref{sec:dqip} we were primarily concerned with the problem of approximating only the value $\lambda(V)$ of the min-max problem \eqref{eq:lambda-consistent}.
This quantity is the verifier's probability of rejection when both provers act optimally; approximating it suffices to prove $\cls{DQIP}\subseteq\cls{PSPACE}$.

However, our result readily extends to the related search problem of \emph{finding} near-optimal strategies for the provers.
Indeed, step \ref{item:optimal-strategies} of the algorithm of \ifthenelse{\equal{\compileCC}{1}}{Section \ref{sec:alg}}{Figure \ref{fig:sqg-alg}} returns a transcript $(\rho_0,\dots,\rho_a)\in\bA(V)$ and a measurement operator $P\in\bP(V)$, both of which are $\delta$-optimal for $\lambda(V)$.
The unitaries $(A_1,\dots,A_a)$ for the yes-prover can be recovered from the transcript $(\rho_0,\dots,\rho_a)$ via the method described in Section \ref{sec:dqip:oracle} with no additional complication.

It is only slightly more difficult to recover the no-prover's unitaries $(B_1,\dots,B_b)$ from $P$.
Our definition of Problem \ref{problem:oracle} (Optimization over $\bP$) specifies only that a solution produce a near-optimal measurement operator $P\in\bP$ for a given state $\rho$.
But the algorithm for Problem \ref{problem:oracle} described in Section \ref{sec:dqip:oracle} for optimization over $\bP(V)$ produces its output $P$ by first constructing the associated unitaries $(B,\dots,B_b)$.
It is a simple matter to modify our definition of Problem \ref{problem:oracle} so as to also return those unitaries in addition to $P$.

The near-optimal measurement operator $P$ returned in step \ref{item:optimal-strategies} of the algorithm of \ifthenelse{\equal{\compileCC}{1}}{Section \ref{sec:alg}}{Figure \ref{fig:sqg-alg}} is given by \[P = \frac{1}{T}\sum_{t=1}^T P^{(t)},\]
which indicates a strategy for the no-prover that selects $t\in\set{1,\dots,T}$ uniformly at random and then acts according to $(B_1^{(t)},\dots,B_b^{(t)})$.
It is a simple matter to construct unitaries $(B_1,\dots,B_b)$ that implement this probabilistic strategy by sampling the integer $t$ during the first round, recording that integer in the no-prover's private memory (which must be enlarged slightly to make room for it), and controlling the operation in subsequent turns on the contents of that integer.
All of the matrix operations required to construct $(B_1,\dots,B_b)$ from each $(B_1^{(t)},\dots,B_b^{(t)})$ in this way can be implemented efficiently in parallel.

\subsection{Robustness with respect to error}
\label{sec:consequences:robust}

In Section \ref{sec:intro:dqip:def} we noted that it is not immediately obvious that the classes $\cls{DIP}$ and $\cls{DQIP}$ are robust with respect to completeness and soundness parameters $c,s$.
Because of this we defined the classes to be inclusive as possible, allowing any verifier for which $c-s\geq 1/p$ for some polynomial-bounded function $p(|x|)$.

Nevertheless, it follows from the collapse of these classes to $\cls{PSPACE}$ that they are indeed robust with respect to completeness and soundness.
In particular, classical interactive proofs for $\cls{PSPACE}$ \cite{LundF+92,Shamir92} imply that if a decision problem $L$ admits a double (quantum) interactive proof with $c-s\geq 1/p$ then $L$ also admits a double (quantum) interactive proof with $c=1$ and $s\leq 2^{-q}$ for any desired polynomial-bounded function $q(|x|)$.

However, the method by which the original verifier is transformed into the low-error verifier is very circuitous: the original verifier must be simulated in polynomial space according to Theorem \ref{thm:dqip-in-pspace} and then that polynomial-space computation must be converted back into an interactive proof with perfect completeness and exponentially small soundness according to proofs of $\cls{IP}=\cls{PSPACE}$.
It would be nice to know whether a more straightforward transformation such as parallel repetition followed by a majority vote could be used to reduce error for double quantum interactive proofs and other bounded-turn interactive proofs with competing provers.

\subsection{Arbitrary payoff observables}

In the study of interactive proofs attention is generally restricted to the \emph{accept-reject} model wherein the verifier's measurement $\set{\Pi,I-\Pi}$ indicates only acceptance or rejection without specifying a payout to the provers.
From a game-theoretic perspective, one might wish to consider a more general verifier whose final measurement $\set{\Pi_a}_{a\in\Sigma}$ could have outcomes belonging to some arbitrary finite set $\Sigma$.
In this case, the verifier awards \emph{payouts} to the provers according to a \emph{payout function} $v:\Sigma\to\mathbb{R}$ where $v(a)$ denotes the payout to the yes-prover in the event of outcome $a$.
(Since the game is zero-sum, the no-prover's payout must be $-v(a)$.)

Jain and Watrous describe a simple transformation by which their algorithm for one-turn quantum games can be used to approximate the expected payout in this more general setting \cite{JainW08}.
Their transformation extends without complication to double quantum interactive proofs.

In our case, the expected payout to the yes-prover when she and the no-prover play according to $(A_1,\dots,A_a)$ and $(B_1,\dots,B_b)$, respectively, is given by
\[
  \sum_{a\in\Sigma} v(a) \bra{\phi} \Pi_a \ket{\phi}
  = \bra{\phi} \Pi_\Sigma \ket{\phi}
\]
where
\[
  \ket{\phi} = B_b V_{a+b-1} B_{b-1}\cdots B_1 V_a A_a V_{a-1} A_{a-1}\cdots A_2 V_1 A_1 \ket{\psi}
\]
is the final state of the system and the Hermitian operator $\Pi_\Sigma=\sum_{a\in\Sigma} v(a) \Pi_a$ denotes the \emph{payout observable} induced by the verifier.
The expected payout of this interaction can be computed simply by translating and rescaling $\Pi_\Sigma$ so as to obtain a measurement operator $0\preceq \Pi \preceq I$ and then running our algorithm for double quantum interactive proofs with verifier $V=(\ket{\psi},V_1,\dots,V_{a+b-1},\Pi)$.
The expected payout of the original protocol is then obtained by inverting the scaling and translation operations by which $\Pi$ was obtained from $\Pi_\Sigma$.
As noted by Jain and Watrous, this transformation has the effect of inflating the additive approximation error $\delta$ by a factor of $\norm{\Pi_\Sigma}$, which is the maximum absolute value of any given payout.

\section*{Acknowledgements}

An extended abstract of this paper has appeared as Ref.\ \cite{GutoskiW12-conf}.
The authors are grateful to Tsuyoshi Ito, Rahul Jain, Zhengfeng Ji, Yaoyun Shi, Sarvagya Upadhyay, John Watrous, and an anonymous reviewer for helpful comments and discussions.
Particularly, the alternative formulation of the
strategies by density operators and measurements is inspired during
the discussion with John Watrous. XW also wants to thank the
hospitality and invaluable guidance of John Watrous when he was
visiting the Institute for Quantum Computing, University of
Waterloo. The research was partially conducted during this visit and
was supported by the Canadian Institute for Advanced Research
(CIFAR). XW's research is also supported by NSF grant 1017335. GG's
research is supported by the Government of Canada through Industry
Canada, the Province of Ontario through the Ministry of Research and
Innovation, NSERC, DTO-ARO, CIFAR, and QuantumWorks.


\begin{thebibliography}{LFKN92}

\bibitem[AHK05]{AroraHK05}
Sanjeev Arora, Elad Hazan, and Satyen Kale.
\newblock The multiplicative weights update method: a meta algorithm and
  applications.
\newblock Submitted, 2005.

\bibitem[Bor77]{Borodin77}
Allan Borodin.
\newblock On relating time and space to size and depth.
\newblock {\em SIAM Journal on Computing}, 6(4):733--744, 1977.

\bibitem[Fan53]{Fan53}
K.~Fan.
\newblock Minimax theorems.
\newblock {\em Proceedings of the National Academy of Sciences}, 39:42--47,
  1953.

\bibitem[FIKU08]{FortnowI+08}
Lance Fortnow, Russell Impagliazzo, Valentine Kabanets, and Christopher Umans.
\newblock On the complexity of succinct zero-sum games.
\newblock {\em Computational Complexity}, 17(3):353--376, 2008.

\bibitem[FK97]{FeigeK97}
Uriel Feige and Joe Kilian.
\newblock Making games short.
\newblock In {\em Proceedings of the 29th ACM Symposium on Theory of Computing
  (STOC 1997)}, pages 506--516, 1997.

\bibitem[FKS95]{FeigenbaumK+95}
Joan Feigenbaum, Daphne Koller, and Peter Shor.
\newblock A game-theoretic classification of interactive complexity classes.
\newblock In {\em Proceedings of the 10th Conference on Structure in Complexity
  Theory}, pages 227--237, 1995.

\bibitem[FvdG99]{FuchsvdG99}
Christopher Fuchs and Jeroen van~de Graaf.
\newblock Cryptographic distinguishability measures for quantum mechanical
  states.
\newblock {\em IEEE Transactions on Information Theory}, 45(4):1216--1227,
  1999.
\newblock arXiv:quant-ph/9712042v2.

\bibitem[GS89]{GoldwasserS89}
Shafi Goldwasser and Michael Sipser.
\newblock Private coins versus public coins in interactive proof systems.
\newblock In Silvio Micali, editor, {\em Randomness and Computation}, volume~5
  of {\em Advances in Computing Research}, pages 73--90. JAI Press, 1989.

\bibitem[GW05]{GutoskiW05}
Gus Gutoski and John Watrous.
\newblock Quantum interactive proofs with competing provers.
\newblock In {\em Proceedings of the 22nd Symposium on Theoretical Aspects of
  Computer Science (STACS'05)}, volume 3404 of {\em Lecture Notes in Computer
  Science}, pages 605--616. Springer, 2005.
\newblock arXiv:cs/0412102v1 [cs.CC].

\bibitem[GW07]{GutoskiW07}
Gus Gutoski and John Watrous.
\newblock Toward a general theory of quantum games.
\newblock In {\em Proceedings of the 39th ACM Symposium on Theory of Computing
  (STOC 2007)}, pages 565--574, 2007.
\newblock arXiv:quant-ph/0611234v2.

\bibitem[GW12]{GutoskiW12-conf}
Gus Gutoski and Xiaodi Wu.
\newblock Parallel approximation of min-max problems with applications to
  classical and quantum zero-sum games.
\newblock In {\em Proceedings of the 27th IEEE Conference on Computational
  Complexity (CCC 2012)}, pages 21--31, 2012.
\newblock arXiv:1011.2787 [quant-ph].

\bibitem[JJUW11]{JainJ+11}
Rahul Jain, Zhengfeng Ji, Sarvagya Upadhyay, and John Watrous.
\newblock {QIP}$=${PSPACE}.
\newblock {\em Journal of the ACM}, 58(6):article 30, 2011.

\bibitem[JUW09]{JainU+09}
Rahul Jain, Sarvagya Upadhyay, and John Watrous.
\newblock Two-message quantum interactive proofs are in {PSPACE}.
\newblock In {\em Proceedings of the 50th IEEE Symposium on Foundations of
  Computer Science (FOCS 2009)}, pages 534--543, 2009.
\newblock arXiv:0905.1300v1 [quant-ph].

\bibitem[JW09]{JainW08}
Rahul Jain and John Watrous.
\newblock Parallel approximation of non-interactive zero-sum quantum games.
\newblock In {\em Proceedings of the 24th IEEE Conference on Computational
  Complexity (CCC 2009)}, pages 243--253, 2009.
\newblock arXiv:0808.2775v1 [quant-ph].

\bibitem[JY11]{JainY11}
Rahul Jain and Penghui Yao.
\newblock A parallel approximation algorithm for positive semidefinite
  programming.
\newblock In {\em Proceedings of the 52nd IEEE Symposium on Foundations of
  Computer Science (FOCS 2011)}, pages 463--471, 2011.
\newblock arXiv:1104.2502v1 [cs.CC].

\bibitem[JY12]{JainY12}
Rahul Jain and Penghui Yao.
\newblock A parallel approximation algorithm for mixed packing and covering
  semidefinite programs.
\newblock arXiv:1201.6090v1 [cs.DS], 2012.

\bibitem[Kal07]{Kale07}
Satyen Kale.
\newblock {\em Efficient algorithms using the multiplicative weights update
  method}.
\newblock PhD thesis, Princeton University, 2007.

\bibitem[Kit02]{Kitaev02}
Alexei Kitaev.
\newblock Quantum coin-flipping.
\newblock Presentation at the 6th Workshop on {\it Quantum Information
  Processing} (QIP 2003), 2002.

\bibitem[KM92]{KollerM92}
Daphne Koller and Nimrod Megiddo.
\newblock The complexity of two-person zero-sum games in extensive form.
\newblock {\em Games and Economic Behavior}, 4:528--552, 1992.

\bibitem[KMvS94]{KollerMvS94}
Daphne Koller, Nimrod Megiddo, and Bernhard von Stengel.
\newblock Fast algorithms for finding randomized strategies in game trees.
\newblock In {\em Proceedings of the 26th ACM Symposium on Theory of Computing
  (STOC 1994)}, pages 750--759, 1994.

\bibitem[KW00]{KitaevW00}
Alexei Kitaev and John Watrous.
\newblock Parallelization, amplification, and exponential time simulation of
  quantum interactive proof system.
\newblock In {\em Proceedings of the 32nd ACM Symposium on Theory of
  Computing}, pages 608--617, 2000.

\bibitem[LFKN92]{LundF+92}
Carsten Lund, Lance Fortnow, Howard Karloff, and Noam Nisan.
\newblock Algebraic methods for interactive proof systems.
\newblock {\em Journal of the ACM}, 39(4):859--868, 1992.

\bibitem[LN93]{LubyN93}
Michael Luby and Noam Nisan.
\newblock A parallel approximation algorithm for positive linear programming.
\newblock In {\em Proceedings of the 25th ACM Symposium on Theory of Computing
  (STOC 1993)}, pages 448--457, 1993.

\bibitem[Meg92]{Megiddo92}
Nimrod Megiddo.
\newblock A note on approximate linear programming.
\newblock {\em Information Processing Letters}, 42(1):53, 1992.

\bibitem[MW05]{MarriottW05}
Chris Marriott and John Watrous.
\newblock Quantum {Arthur-Merlin} games.
\newblock {\em Computational Complexity}, 14(2):122--152, 2005.
\newblock arXiv:cs/0506068v1 [cs.CC].

\bibitem[NC00]{NielsenC00}
Michael Nielsen and Issac Chuang.
\newblock {\em Quantum Computation and Quantum Information}.
\newblock Cambridge University Press, 2000.

\bibitem[Pap94]{Papadimitriou94}
Christos Papadimitriou.
\newblock {\em Computational Complexity}.
\newblock Addison-Wesley, 1994.

\bibitem[PT12]{PengT12}
Richard Peng and Kanat Tangwongsan.
\newblock Faster and simpler width-independent parallel algorithms for positive
  semidefinite programming.
\newblock In {\em Proceedings of the 24th ACM symposium on Parallelism in
  algorithms and architectures (SPAA 2012)}, pages 101--108, 2012.
\newblock arXiv:1201.5135 [cs.DS].

\bibitem[RW05]{RosgenW05}
Bill Rosgen and John Watrous.
\newblock On the hardness of distinguishing mixed-state quantum computations.
\newblock In {\em Proceedings of the 20th Conference on Computational
  Complexity}, pages 344--354, 2005.
\newblock arXiv:cs/0407056v1 [cs.CC].

\bibitem[Ser91]{Serna91}
Maria Serna.
\newblock Approximating linear programming is log-space complete for {P}.
\newblock {\em Information Processing Letters}, 37(4):233--236, 1991.

\bibitem[Sha92]{Shamir92}
Adi Shamir.
\newblock {IP} $=$ {PSPACE}.
\newblock {\em Journal of the ACM}, 39(4):869--877, 1992.

\bibitem[TX98]{TrevisanX98}
Luca Trevisan and Fatos Xhafa.
\newblock The parallel complexity of positive linear programming.
\newblock {\em Parallel Processing Letters}, 8(4):527--533, 1998.

\bibitem[vN28]{vonNeumann28}
John von Neumann.
\newblock Zur theorie der gesellschaftspiele.
\newblock {\em Mathematische Annalen}, 100(1):295--320, 1928.
\newblock In German.

\bibitem[vzG93]{vzGathen93}
Joachim von~zur Gathen.
\newblock Parallel linear algebra.
\newblock In John~H. Reif, editor, {\em Synthesis of Parallel Algorithms},
  chapter~13. Morgan Kaufmann Publishers, Inc., 1993.

\bibitem[Wat11]{Watrous11-lec}
John Watrous.
\newblock Lecture notes: Theory of quantum information.
\newblock Available on the author's web page, 2011.

\bibitem[WK06]{WarmuthK06}
Manfred Warmuth and Dima Kuzmin.
\newblock Online variance minimization.
\newblock In {\em Proceedings of the 19th Conference on Learning Theory},
  volume 4505 of {\em Lecture Notes in Computer Science}, pages 514--528, 2006.

\bibitem[Wu10]{Wu10}
Xiaodi Wu.
\newblock Equilibrium value method for the proof of {QIP}$=${PSPACE}.
\newblock arXiv:1004.0264v2 [quant-ph], 2010.

\bibitem[You01]{Young01}
Neal Young.
\newblock Sequential and parallel algorithms for mixed packing and covering.
\newblock In {\em Proceedings of the 42nd IEEE Symposium on Foundations of
  Computer Science (FOCS 2001)}, pages 538--546, 2001.

\end{thebibliography}
\end{document}